\documentclass[11pt]{article}
\RequirePackage{algorithm,algorithmic,amssymb,epsf,epic,amsmath}

\def\denseformat{
\setlength{\textheight}{9in}
\setlength{\textwidth}{6.9in}
\setlength{\evensidemargin}{-0.2in}
\setlength{\oddsidemargin}{-0.2in}
\setlength{\headsep}{10pt}
\setlength{\topmargin}{-0.3in}
\setlength{\columnsep}{0.375in}
\setlength{\itemsep}{0pt}
}

\newtheorem{theorem}{Theorem}[section]

\newtheorem{claim}[theorem]{Claim}
\newtheorem{lemma}[theorem]{Lemma}

\newtheorem{corollary}[theorem]{Corollary}

\newtheorem{comment}[theorem]{Comment}
\newtheorem{remark}[theorem]{Remark}

\def\boldhead#1:{\par\vskip 7pt\noindent{\bf #1:}\hskip 10pt}
\def\ithead#1:{\par\vskip 7pt\noindent{\it #1:}\hskip 10pt}

\def\inline#1:{\par\vskip 7pt\noindent{\bf #1:}\hskip 10pt}
\def\midinline#1:{\par\noindent{\bf #1:}\hskip 10pt}
\def\dnsinline#1:{\par\vskip -7pt\noindent{\bf #1:}\hskip 10pt}
\def\ddnsinline#1:{\newline{\bf #1:}\hskip 10pt}
\def\largeinline#1:{\par\vskip 7pt\noindent{\large\bf #1:}\hskip 10pt}
\long\def\comment #1\commentend{}
\long\def\commhide #1\commhideend{}
\long\def\commfull #1\commend{#1}
\long\def\commabs #1\commenda{}
\long\def\commtim #1\commendt{#1}
\long\def\commb #1\commbend{}
\long\def\commedit #1\commeditend{} 

\long\def\commB #1\commBend{}       

\long\def\commex #1\commexend{}     

\long\def\commsiena #1\commsienaend{}  
                                         
\long\def\commBI #1\commBIend{}  
                                         
\long\def\CProof #1\CQED{}

\def\blackslug{\hbox{\hskip 1pt \vrule width 4pt height 8pt
    depth 1.5pt \hskip 1pt}}
\def\QED{\quad\blackslug\lower 8.5pt\null\par}

\def\Proof{\par\noindent{\bf Proof:~}}

\def\proof{\Proof}

\long\def\PPP#1{\noindent{\bf Proof:}{ #1}{\quad\blackslug\lower 8.5pt\null}}

\long\def\denspar #1\densend
{#1}

\newif\ifnotesw\noteswtrue
   {\ifnotesw\marginpar[\hfill\(\top\)]{\(\top\)}\fi}%
   {\ifnotesw\marginpar[\hfill\(\bot\)]{\(\bot\)}\fi}

\newcommand{\mnote}[1]%
    {\ifnotesw\marginpar%
        [{\scriptsize\it\begin{minipage}[t]{\marginparwidth}
        \raggedleft#1%
                        \end{minipage}}]%
        {\scriptsize\it\begin{minipage}[t]{\marginparwidth}
        \raggedright#1%
                        \end{minipage}}%
    \fi}

\def\cB{{\cal B}}
\def\cC{{\cal C}}

\def\cL{{\cal L}}

\def\cP{{\cal P}}

\def\cS{{\cal S}}

\def\cW{{\cal W}}

\def\hC{{\hat C}}

\def\hL{{\hat L}}

\def\hR{{\hat R}}

\def\tE{{\tilde E}}

\def\tG{{\tilde G}}

\def\tO{{\tilde O}}

\def\tV{{\tilde V}}

\def\MathF{\hbox{\rm I\kern-2pt F}}
\def\MathP{\hbox{\rm I\kern-2pt P}}
\def\MathR{\hbox{\rm I\kern-2pt R}}
\def\MathZ{\hbox{\sf Z\kern-4pt Z}}
\def\MathN{\hbox{\rm I\kern-2pt I\kern-3.1pt N}}
\def\MathC{\hbox{\rm \kern0.7pt\raise0.8pt\hbox{\footnotesize I}
\kern-4.2pt C}}
\def\MathQ{\hbox{\rm I\kern-6pt Q}}

\def\MathE{\hbox{{\rm I}\hskip -2pt {\rm E}}} 

\newsavebox{\ttop}\newsavebox{\bbot}

\def\eps{\epsilon}

\def\polylog{\mbox{polylog}}
\def\poly{\mbox{poly}}
\def\setmns{\setminus}

\newcommand{\half}{{1\over2}}

\def\nin{{~\not \in~}}
\def\emset{\emptyset}

\newcommand{\Prob}{\MathP}
\newcommand{\Expect}{\MathE}

\def\etal{\emph{et~al.}}

\denseformat
\usepackage{multirow}
\usepackage{boldline}

\usepackage{pdfsync,fullpage,comment}
\usepackage{authblk}

\ifx\pdftexversion\undefined
\usepackage[colorlinks,linkcolor=black,filecolor=black,citecolor=black,urlco
lor=black,pdfstartview=FitH]{hyperref}
\else
\usepackage[colorlinks,linkcolor=blue,filecolor=blue,citecolor=blue,urlcolor
=blue,pdfstartview=FitH]{hyperref}
\fi

\newcommand{\alert}[1]{\textbf{\color{red}
[[[#1]]]}\marginpar{\textbf{\color{red}**}}\typeout{ALERT:
\the\inputlineno: #1}}

\newcommand{\namedref}[2]{\hyperref[#2]{#1~\ref*{#2}}}
\newcommand{\sectionref}[1]{\namedref{Section}{#1}}

\newcommand{\theoremref}[1]{\namedref{Theorem}{#1}}

\newcommand{\claimref}[1]{\namedref{Claim}{#1}}
\newcommand{\lemmaref}[1]{\namedref{Lemma}{#1}}

\newcommand{\remarkref}[1]{\namedref{Remark}{#1}}
\newcommand{\corollaryref}[1]{\namedref{Corollary}{#1}}

\begin{document}
\def\hpi{\hat{\pi}}
\def\rt{\mathit{rt}}
\def\hd{\hat{d}}
\def\chS{\hat{\cal S}}
\def\chP{\hat{\cal P}}
\def\chU{\hat{\cal U}}
\def\exp{\mathit{exp}}
\def\Rt{\mathit{Roots}}
\def\Lab{\mathit{Label}}
\def\Id{\mathit{Id}}
\def\Lev{\mathit{Level}}
\def\URt{\mathit{URoots}}
\def\Ball{\mathit{Ball}}
\def\wmax{{w_{max}}}
\def\tO{\tilde{O}}
\def\nin{\not \in}
\def\emset{\emptyset}
\def\setmns{\setminus}
\def\etal{{et al.~}}
\def\Pairs{\mathit{Pairs}}
\def\Paths{\mathit{Paths}}
\def\pred{\mathit{pred}}
\def\Rad{\mathit{Rad}}
\def\succ{\mathit{succ}}
\def\NULL{\mathit{NULL}}
\def\exp{\mathit{exp}}
\def\Ball{\mathit{Ball}}
\def\tPi{\tilde{\Pi}}
\def\deg{\mathit{deg}}
\def\TB{\cB^{{1/3}}}
\def\Branch{\mathit{Branch}}
\def\uzero{u^{(0)}}
\def\uone{u^{(1)}}
\def\uj{u^{(j)}}
\def\vzero{v^{(0)}}
\def\vone{v^{(1)}}
\def\vj{v^{(j)}}
\def\dzero{d^{(0)}}
\def\done{d^{(1)}}
\def\dj{d^{(j)}}
\def\dG{d_G}
\def\third{{1 \over 3}}
\def\stretchexp{{\log_{4/3} 7}}
\def\tLambda{\tilde{\Lambda}}
\def\tomega{\tilde{\omega}}
\def\HKNfactor{ 2^{\tO(\sqrt{\log \max \{n,\Lambda\}})}}
\def\ourfactor{ (1/\eps)^{O(\sqrt{{\log n} \over {\log\log n}})} \cdot 2^{O(\sqrt{\log n \cdot \log\log n})} }

\newcommand{\Patrascu}{P\v{a}tra\c{s}cu{~}}
\newcommand{\Lists}{{\rm Lists}}
\newcommand{\td}{{\tilde{d}}}

\title{Hopsets with Constant Hopbound, and Applications to Approximate Shortest Paths}
\author[1]{Michael Elkin\thanks{This research was supported by the ISF grant 724/15.}}
\author[1]{Ofer Neiman\thanks{Supported in part by ISF grant No. (523/12) and by the European Union Seventh Framework
Programme (FP7/2007-2013) under grant agreement $n^\circ 303809$.}}

\affil[1]{Department of Computer Science, Ben-Gurion University of the Negev,
Beer-Sheva, Israel. Email: \texttt{\{elkinm,neimano\}@cs.bgu.ac.il}}

\date{}
\maketitle

\begin{abstract}
A $(\beta,\eps)$-hopset for a weighted undirected $n$-vertex graph $G=(V,E)$ is a set of edges, whose addition to the graph guarantees that every pair of vertices has a path between them that contains at most $\beta$ edges, whose length is within $1+\eps$ of the shortest path. In her seminal paper, Cohen \cite[JACM 2000]{C00} introduced the notion of hopsets in the context of parallel computation of approximate shortest paths, and since then it has found numerous applications in various other settings, such as dynamic graph algorithms, distributed computing, and the streaming model.

Cohen \cite{C00} devised efficient algorithms for constructing hopsets with {\em polylogarithmic} in  $n$ number of hops. Her constructions remain the state-of-the--art since the publication of her paper in STOC'94, i.e., for more than two decades.

In this paper we exhibit the first construction of sparse hopsets with a {\em constant number of hops}. We also find efficient algorithms for hopsets in various computational settings, improving the best known constructions.
Generally, our hopsets strictly outperform the hopsets of \cite{C00}, both in terms of their parameters, and in terms of the resources required to construct them.

We demonstrate the applicability of our results for the fundamental problem of computing approximate shortest paths from $s$ sources. Our results improve the running time for this problem in the parallel, distributed and streaming models, for a vast range of $s$.
\end{abstract}


\thispagestyle{empty}
\newpage
\setcounter{page}{1}

\section{Introduction}

\subsection{Hopsets, Setting and Main Results}

We are given an $n$-vertex weighted undirected graph $G = (V,E,\omega)$. Consider another graph $G_H = (V,H,\omega_H)$ on the same vertex set $V$.
Define the union graph $G' = G \cup G_H$,
$G' = (V,E' = E \cup H,\omega')$, where $\omega'(e) = \omega_H(e)$ for $e \in H$, and $\omega'(e) = \omega(e)$ for $e \in E \setminus H$. For a positive  integer parameter $\beta$, and a pair $u,v \in V$ of distinct vertices, a {\em $\beta$-limited distance} between $u$ and $v$ in $G'$, denoted $d_{G'}^{(\beta)}(u,v)$, is the length of the shortest $u$-$v$ path in $G'$ that contains at most $\beta$ edges (aka {\em hops}). For a parameter $\eps > 0$, and a positive integer $\beta$ as above, a graph $G_H = (V,H,\omega_H)$ is called a {\em $(\beta,\eps)$-hopset} for the graph $G$, if for every pair $u,v \in V$ of vertices, we have
$d_G(u,v) \le d_{G'}^{(\beta)}(u,v) \le (1 + \eps) \cdot d_G(u,v)$.
 (Here $d_G(u,v)$ stands for the distance between $u$ and $v$ in $G$.)
We often refer to the edge set $H$ of $G_H$ as the {\em hopset}. The parameter $\beta$ is called the {\em hopbound} of the hopset.

Hopsets are a  fundamental  graph-algorithmic construct. They turn out extremely useful for computing approximate shortest paths, distances, and for routing problems in numerous computational settings, in which computing shortest paths with a limited number of hops is significantly easier than computing shortest paths with no limitation on the number of hops. A partial list of these settings includes distributed, parallel, streaming and centralized dynamic models.

Hopsets were explicitly introduced in Cohen's seminal STOC'94 paper \cite{C00}. Implicit constructions of hopsets were given already in the beginning of nineties by Ullman and Yannakakis \cite{UY91}, Klein and Sairam \cite{KS97}, Cohen \cite{C97}, and Shi and Spencer \cite{SS99}. Cohen \cite{C00} showed that for any parameters $\eps > 0$ and $\kappa = 1,2,\ldots$, and any $n$-vertex graph $G$, there exists a $(\beta,\eps)$-hopset $H$ with $|H|  = \tO(n^{1+1/\kappa})$ edges,\footnote{The notation $\tO(f(n))$ stands for $O(f(n) \cdot \log^{O(1)} f(n))$.} where the hopbound $\beta$ is polylogarithmic in $n$. Specifically, it is given by $\beta = \left({{\log n} \over \eps}\right)^{O(\log \kappa)}$. Algorithmically, she showed that given an additional parameter $\rho > 0$, $(\beta,\eps)$-hopsets with
\begin{equation}
\label{eq:beta_cohen}
\beta_{Coh} = \left({{\log n} \over \eps}\right)^{O((\log \kappa)/\rho)}
\end{equation}
can be computed in $O(|E| \cdot n^\rho)$ time in the centralized model of computation, and in $O(\beta) \cdot \polylog(n)$ PRAM time, with $O(|E| \cdot n^\rho)$ work.
She used these hopsets' constructions to devise efficient parallel algorithms for computing $S \times V$ $(1+\eps)$-approximate shortest paths (henceforth, $(1+\eps)$-ASP). Her results for these problems remained the state-of-the-art in this context up until now, for over two decades.

Despite being a major breakthrough in the nineties, Cohen's hopsets leave much to be desired. Indeed, the only general lower bound applicable to them is that of \cite{CG06} (based on \cite{Y82,AS87}), asserting that there exist $n$-vertex graphs for which any $(\beta,\eps)$-hopset requires $\Omega(n \cdot \log^{(\lfloor \beta/2 \rfloor)} n)$ edges, where $\log^{(t)} n$ stands for a $t$-iterated logarithm.
Cohen \cite{C00} herself wrote in the introduction of her paper (the italics are in the origin):
\\
\\
``One intriguing issue is the {\em existence} question of sparse hop sets with certain attributes. In addition, we would like to construct them efficiently."
\\
\\
The same motive repeats itself in the concluding section of her paper, where she writes:
\\
\\
``We find the existence of good hop sets to be an intriguing research problem on its own right."
\\
\\
In the more than twenty years that passed since Cohen's \cite{C00} paper was published (in STOC'94), numerous additional applications of hopsets were discovered, and also some new constructions of hopsets were presented. Most notably, Bernstein \cite{B09} and Henzinger \etal \cite{HKN14} devised new constructions of hopsets, and used them for maintaining approximate shortest paths in dynamic centralized setting. Nanongkai \cite{N14} and Henzinger \etal \cite{HKN15} used hopsets for computing approximate shortest paths in distributed and streaming settings. Lenzen and Patt-Shamir \cite{LP15} and the authors of the current paper \cite{EN16} used them for compact routing. Miller \etal \cite{MPVX15} devised  new constructions of hopsets, and used them for approximate shortest paths in PRAM setting.
The hopsets of \cite{B09,HKN14,N14,HKN15} have hopbound $2^{\tO(\sqrt{\log n})}$, and size $n \cdot 2^{\tO(\sqrt{\log n})}\cdot\log\Lambda$, where $\Lambda$ is the aspect ratio of the graph.\footnote{The aspect ratio of a graph $G$ is defined by the ratio of the largest distance to the smallest distance in $G$.} The hopsets of Miller \etal \cite{MPVX15} have hopbound at least $\Omega(n^{\alpha})$, for a constant $\alpha > 0$, and linear size.

We will discuss these results of \cite{B09,HKN14,N14,HKN15,MPVX15} in greater detail in the sequel. However, they all fail to address the fundamental challenge of Cohen \cite{C00}, concerning the existence and efficient constructability of hopsets that are strictly and substantially superior to those devised in \cite{C00}.
In this paper we build such hopsets. Specifically, for any $\eps >0$, $\kappa = 1,2\ldots$, and any $n$-vertex graph $G = (V,E)$, we show that there exists a $(\beta,\eps)$-hopset with $\beta = \left( {{\log \kappa} \over \eps} \right)^{\log \kappa}$, and $O(n^{1+1/\kappa} \cdot \log n)$ edges. Hence, these hopsets {\em simultaneously}  exhibit arbitrarily small constant approximation factor $1+\eps$, arbitrarily close to 1 constant exponent $1+ 1/\kappa$ of the hopset's size, and {\em constant} hopbound $\beta$.
In all previous hopsets' constructions, the hopbound was at least {\em polylogarithmic} in $n$, in all regimes. Moreover, we devise efficient algorithms to build our hopsets in various computational models. Specifically, given a parameter $\rho > 0$ that controls the running time, our centralized algorithm constructs a $(\beta,\eps)$-hopset with $O(n^{1+1/\kappa} \cdot \log n)$ edges in expected $O(|E|\cdot n^\rho)$ time, with
\begin{equation}
\label{eq:our_beta}
\beta ~=~ O\left({1 \over \eps} \cdot (\log \kappa  + 1/\rho)\right)^{\log \kappa  + O(1/\rho)}~.
\end{equation}
Again, we can simultaneously have arbitrarily small constant approximation $1+\eps$, arbitrarily close to 1 hopset's size exponent $1 + 1/\kappa$, arbitrarily close to $O(|E|)$ running time (in the sense $O(|E| \cdot n^\rho)$, for an arbitrarily small constant $\rho > 0$), and still the hopbound $\beta$ of the constructed hopset remains {\em constant}!

As was mentioned above, in \cite{C00} (the previous state-of-the-art), with the same approximation factor, hopset size and running time, the hopbound behaves as given in (\ref{eq:beta_cohen}).
Hence our result is stronger than that of \cite{C00} in a number of senses. First, the hopbound $\beta_{Coh}$ is at least polylogarithmic, while ours is constant.
Second, the exponent $O((\log \kappa)/\rho)$ of $\beta_{Coh}$ is substantially larger than the exponent $\log \kappa + O(1/\rho)$ in our $\beta$. See Table \ref{table:hopset} for a concise comparison of existing hopsets' constructions.

\begin{table*}
\centering

\begin{tabular}{ |l|l|l|l| }
\hline
 Reference & Size & Hopbound & Run-time\\
\hline
\cite{B09,HKN14,HKN15} & $n\cdot 2^{\tilde{O}(\sqrt{\log n})} \cdot \log \Lambda $ & $2^{\tilde{O}(\sqrt{\log n})}$ & \\
\hline
\cite{MPVX15} & $O(n)$ & $n^\alpha$, ($\alpha = \Omega(1)$) & \\
\hline
\cite{C00} & $O(n^{1+\frac{1}{\kappa}}\cdot\log n)$ &  $(\log n)^{O(\frac{\log\kappa}{\rho})}$ & $|E|\cdot n^\rho$ \\
\hlineB{3}
{\bf This paper} & $O(n^{1+\frac{1}{\kappa}}\cdot\log n)$ &  $\left(\log\kappa+\frac{1}{\rho}\right)^{\log\kappa+O(\frac{1}{\rho})}$ & $|E|\cdot n^\rho$ \\
\hline
\end{tabular}
\caption{Comparison between $(\beta,\eps)$-hopsets (neglecting the dependency on $\eps$). We note that the hopsets of \cite{B09,HKN14,HKN15,MPVX15} were designed for certain computational models (i.e., dynamic, streaming, distributed).}\label{table:hopset}
\end{table*}

\subsection{Hopsets in Parallel, Streaming and Distributed Models}

We also devise efficient parallel, distributed and streaming algorithms for constructing hopsets with constant hopbound.

\subsubsection{Hopsets in the Streaming Model}

In the streaming model, the only previously known algorithm for constructing hopsets is that of \cite{HKN15}. Using $2^{\tO(\sqrt{\log n})} \cdot \log \Lambda$ passes over the stream, and $n \cdot 2^{\tO(\sqrt{\log n})} \cdot \log \Lambda$ space, their algorithm produces a hopset with hopbound $\beta = 2^{\tO(\sqrt{\log n})}$, and size $n \cdot 2^{\tO(\sqrt{\log n})}\cdot\log\Lambda$. Our streaming algorithm constructs a hopset with $\beta$ \begin{equation}
\label{eq:beta_streaming}
\beta = O\left({{\log \kappa  + 1/\rho} \over {\eps \cdot \rho}}\right)^{\log \kappa  + O(1/\rho)}~,
\end{equation}
i.e., it is independent of $n$.
The expected size of the hopset is $O(n^{1+1/\kappa} \cdot \log n)$
, and
it uses space $O(n^{1+1/\kappa} \cdot \log^2 n)$
. The number of passes is $O(n^\rho \cdot \beta)$
.
Also, by setting $\kappa = \Theta(\log n)$, $\rho = \sqrt{{\log\log n} \over {\log n}}$, our result  strictly dominates that of \cite{HKN15}; the hopbound and number of passes are essentially the same, while our space usage and hopset's size are significantly better.

\subsubsection{Hopsets in the PRAM Model}

In the PRAM model, Klein and Sairam \cite{KS97} and Shi and Spencer \cite{SS99} (implicitly) devised algorithms for constructing exact ($\eps = 0$) hopsets with hopbound $\beta = O(\sqrt{n})$ of linear size $O(n)$, in parallel time $O(\sqrt{n} \cdot \log n)$, and with $O(|E| \cdot \sqrt{n})$ work.
(Work is the total number of operations performed by all processors during the algorithm.)
Cohen \cite{C00} constructed $(\beta,\eps)$-hopsets with size
 $n^{1+1/\kappa} \cdot (\log n)^{O((\log \kappa)/ \rho)}$, with hopbound $\beta_{Coh}$ given by (\ref{eq:beta_cohen}),
in parallel time $\left({{\log n} \over \eps}\right)^{O((\log \kappa)/\rho)}$, using $O(|E| \cdot n^\rho)$ work.
Her $\kappa$ and $\rho$ are restricted by $\kappa,1/\rho = O(\log\log n)$, and thus the resulting hopset is never sparser than $n \cdot 2^{O({{\log n} \over {\log \log n}})}$.

Miller \etal \cite{MPVX15} devised two constructions of linear-size  $(\beta,\eps)$-hopsets, but with very large $\beta$. One has $\beta = O_\eps(n^{{4 + \alpha} \over {4 + 2\alpha}})$, and  running time given by the same expression, and work $O(|E| \cdot \log^{3 + \alpha} n)$, for a free parameter $\alpha$. Another has $\beta = n^\alpha$, for a constant $\alpha$, and running time given by the same expression, and work $O(|E| \cdot \log^{O(1/\alpha)}n)$.

Our algorithm has two regimes. In the first regime it constructs $(\beta,\eps)$-hopsets with
$\beta = \left({{\log n} \over \eps} \right)^{\log \kappa + O(1/\rho)}$, with expected size $O(n^{1+1/\kappa} \cdot \log n)$, in time
$\left({{\log n} \over \eps} \right)^{\log \kappa + O(1/\rho)}$, using $O(|E| \cdot n^\rho)$ work.
This result strictly improves upon Cohen's hopset \cite{C00}, as the exponent of $\beta$ and of the running time in the latter is $O((\log \kappa)/\rho)$, instead of $\log \kappa  + O(1/\rho)$ in our case. Also, the size of our hopset is smaller than that of \cite{C00} by a factor of $O(\log n)^{O((\log \kappa)/\rho)}$.

In the second regime our PRAM algorithm computes a hopset with constant (i.e., independent of $n$)  hopbound $\beta$, but in larger parallel time.
See Table \ref{table:pram} for a concise comparison of available PRAM algorithms.

\begin{table*}
\centering
\small
\begin{tabular}{ |l|l|l|l|l| }
\hline
 Reference & Size & $\beta$ = Hopbound & Time & Work \\
\hline
\cite{KS97,SS99} & $O(n)$ & $O(\sqrt{n})$ & $O(\sqrt{n}\log n)$ & $O(|E|\cdot\sqrt{n})$\\
\hlineB{2}
\multirow{2}{*}{\cite{MPVX15}} & $O(n)$ & $O(n^{\frac{4+\alpha}{4+2\alpha}})$ & $O(n^{\frac{4+\alpha}{4+2\alpha}})$ & $O(|E|\cdot\log^{3+\alpha}n)$\\
\cline{2-5}
& $O(n)$ & $O(n^\alpha)$ ($\alpha\ge\Omega(1)$)& $O(n^{\alpha})$ & $O(|E|\cdot\log^{O(1/\alpha)}n)$\\
\hlineB{2}
\cite{C00} & $n^{1+1/\kappa}\cdot (\log n)^{O(\frac{\log\kappa}{\rho})}$ & $(\log n)^{O(\frac{\log\kappa}{\rho})}$ & $(\log n)^{O(\frac{\log\kappa}{\rho})}$  &  $O(|E|\cdot n^\rho)$ \\
\hlineB{3}
\multirow{2}{*}{{\bf This paper} }& $O(n^{1+\frac{1}{\kappa}}\cdot \log n)$ & $\left(\log n\right)^{\log\kappa+O(\frac{1}{\rho})}$ & $(\log n)^{\log\kappa+O(\frac{1}{\rho})}$  & $O(|E|\cdot n^\rho)$ \\
\cline{2-5}
& $O(n^{1+\frac{1}{\kappa}}\cdot\log n)$ & $\left(\frac{\log\kappa+\frac{1}{\rho}}{\zeta}\right)^{\log\kappa+O(\frac{1}{\rho})}$ & $O(n^\zeta)\cdot\beta$  & $O(|E|\cdot n^{\rho+\zeta})$ \\
\hline

\end{tabular}
\caption{Comparison between $(\beta,\eps)$-hopsets in the PRAM model (neglecting the dependency on $\eps$). The hopsets of \cite{KS97,SS99} provide exact distances.}\label{table:pram}
\end{table*}

\subsection{Hopsets in Distributed Models}\label{sec:models}

There are two distributed models in which hopsets were studied in the literature \cite{HKN14,N14,HKN15,LP15,EN16}.
These are the Congested Clique model, and the CONGEST  model. In both models every vertex of an $n$-vertex graph $G = (V,E)$ hosts a processor, and the processors communicate with one another in discrete rounds, via short messages. Each message is allowed to contain an identity of a vertex or an edge, and an edge weight, or anything else of no larger (up to a fixed constant factor) size.\footnote{Typically, in the CONGEST model only messages of size $O(\log n)$ bits are allowed, but edge weights are restricted to be at most polynomial in $n$. Our definition is geared to capture a more general situation, when there is no restriction on the aspect ratio. Hence results achieved in our more general model are more general than previous ones.} On each round each vertex can send possibly different messages to its neighbors. The local computation is assumed to require zero time, and we are interested in algorithms that run for as few rounds as possible. (The number of rounds is called the {\em running time}.)
In the Congested Clique model, we assume that all vertices are interconnected via direct edges, but there might be some other weighted undirected graph $G' = (V,E',\omega)$,  $E' \subseteq E = {V \choose 2}$, embedded in the clique $G$, for which we want to compute a hopset. In the CONGEST model, every vertex can send messages only to its $G$-neighbors, but we also assume that there is an embedded ``virtual" graph $G' = (V',E',\omega)$, $V' \subseteq V$, known locally to the vertices. (Every vertex $u \in V$ knows at the beginning of the computation if $u \in V'$, and if it is the case, then  it also knows the identities of its $G'$-neighbors.) We remark that the assumption of embedded graph $G'$ in the CONGEST model appears in previous papers on computing hopsets in distributed setting, that is, in \cite{HKN14,N14,HKN15,LP15,EN16}. It is motivated by  distributed applications of hopsets, i.e., approximate shortest paths computation, distance estimation and routing, which require a hopset for a virtual graph embedded in the underlying network in the above way.

Henzinger \etal \cite{HKN15} devised an algorithm for constructing hopsets in the Congested Clique model. Their hopset has hopbound $\beta = 2^{\tO(\sqrt{\log n})}$, and size $n \cdot 2^{\tO(\sqrt{\log n})}\cdot\log\Lambda$, where $\Lambda$ is the aspect ratio of the embedded graph.
The running time of their algorithm is $2^{\tO(\sqrt{\log n})}\cdot\log\Lambda$.

Our algorithm, for parameters $\eps > 0$, $\rho > 0$, $\kappa = 2,3,\ldots$, computes a hopset with
\begin{equation}
\label{eq:beta_distr}
\beta = O\left({{\log \kappa  + 1/\rho} \over {\eps \cdot \rho}}\right)^{\log \kappa  + O(1/\rho)}~,
\end{equation}
with expected size $O(n^{1+1/\kappa} \cdot \log n)$, in $O(n^\rho \cdot \beta^2)$ rounds.

Comparing our result to that of \cite{HKN15}, we first note that our hopset achieves a constant (i.e., independent of $n$) hopbound. Second, by setting $\kappa = \Theta(\log n)$, $\rho = \sqrt{{\log\log n} \over {\log n}}$, we can have our hopbound and running time equal to $2^{\tO(\sqrt{\log n })}$, i.e., roughly the same as, but in fact, slightly better than, the respective bounds of \cite{HKN15}. Our hopset's size becomes then $O(n \cdot \log n)$, i.e., much closer to linear than
$n \cdot 2^{\tO(\sqrt{\log n })}$ of the hopset of \cite{HKN15}.

The situation is similar in the CONGEST model. Denote by $m = |V'|$ the size of the vertex set of the embedded graph $G'$. The algorithm of \cite{HKN15} computes a hopset with the same hopbound and size as in the Congested Clique model (with $n$ replaced by $m$), and it does so in $(D + m) \cdot 2^{\tO(\sqrt{\log m})}\cdot\log\Lambda$ time, where $\Lambda$ is the aspect ratio of $G'$, and $D$ is the hop-diameter of $G$.\footnote{The hop-diameter of a graph is the maximum hop-distance between two vertices. The hop-distance  between a pair $u,v$ of vertices is the minimal number of hops in a path between them.} Our algorithm computes a hopset with (constant) hopbound given by (\ref{eq:beta_distr}), expected size $O(m^{1+1/\kappa} \cdot \log m)$, in $O((D + m^{1+\rho}) \cdot \beta \cdot m^\rho)$ time. (See Corollary \ref{cor:congest_hopset_reduction}, which gives, in fact, stronger, but more complicated bounds.) Again, our hopset can have constant hopbound, while that of \cite{HKN15} is $2^{\tO(\sqrt{\log m})}$. Also, by setting $\kappa = \Theta(\log m)$, $\rho = \sqrt{{\log\log m} \over {\log m}}$, we obtain a result, which strictly dominates that of \cite{HKN15}.

\subsection{Applications}

Our algorithms for constructing hopsets also give rise to improved algorithms for the problems of computing $(1+\eps)$-approximate shortest distances
(henceforth, $(1+\eps)$-ASD) and paths (henceforth, $(1+\eps)$-ASP).  In all settings, we consider a subset $S \subseteq V$ of origins, and we are interested in distance estimates or in approximate shortest paths for pairs in $S \times V$.  Denote $s = |S|$.

Our PRAM algorithm for the $(1 + \eps)$-ASP problem has running time
 $O\left({{ \log n} \over \eps}\right)^{\log \kappa + O(1/\rho)}$,
and uses $O(|E| \cdot (n^\rho + s)$ work.  Cohen's algorithm \cite{C00} for the same problem has (parallel) running time $O\left({{\log n} \over \eps}\right)^{O((\log \kappa)/\rho)}$, and has the same work complexity as our algorithm. Hence, both our and Cohen's algorithms achieve polylogarithmic time and near-optimal work complexity, but the exponent of the logarithm in our result is significantly smaller than in Cohen's one.

\begin{comment}
In the congested clique model, Henzinger \etal \cite{HKN15} used hopsets to come up with an algorithm that computes single-source $(1+\eps)$-ASP in $2^{\tO(\sqrt{\log n})}$ time. Applying their algorithm separately from each source results in time $s \cdot 2^{\tO(\sqrt{\log n})}$. Our algorithm computes $S \times V$ $(1+\eps)$-ASP for $s = n^{\Omega(1)}$, in $s \cdot (1/\eps)^{O(1)}$ time.  We remark that an algorithm of Censor-Hillel \etal \cite{CH??} computes all-pairs ASP in $O(n^{0.158})$ time.
Hence our result here improves the state-of-the-art for the range $n^{\Omega(1)} = s = o(n^{0.158})$.
\end{comment}

In the distributed CONGEST model (see \sectionref{sec:models} for its definition), the hopset-based algorithm of \cite{HKN15} computes single-source $(1+\eps)$-ASP in $(D+ \sqrt{n}) \cdot 2^{\tO(\sqrt{\log n})}$ time.
Using it naively for $S \times V$ $(1+\eps)$-ASP results in running time of $(D+ s \cdot \sqrt{n}) \cdot 2^{\tO(\sqrt{\log n})}$. Using our hopsets we solve this problem in
$(D + \sqrt{n \cdot s})  \cdot 2^{\tO(\sqrt{\log n })}$ time.
Whenever $s = n^{\Omega(1)}$, we use our hopset with different parameters, and our running time becomes $\tO(D + \sqrt{n \cdot s})$.

In the streaming model, Henzinger \etal \cite{HKN15} devised a single-source $(1 + \eps)$-ASP streaming algorithm with $2^{\tO(\sqrt{\log n})} \cdot \log \Lambda$ passes, that uses $n \cdot 2^{\tO(\sqrt{\log n})} \cdot \log \Lambda$ space. To the best of our knowledge, the best-known streaming $S \times V$ $(1+\eps)$-ASP algorithm with this space requirement is to run the algorithm of \cite{HKN15} for each source separately, one after another. The resulting number of passes is
$s \cdot 2^{\tO(\sqrt{\log n})} \cdot \log \Lambda$. Our algorithm for this problem builds a hopset, whose parameters depend on $s$. As a result, our algorithm  has an  improved number of passes, particularly when $s$ is large (we also avoid the dependence on $\Lambda$). Our space usage is only $\tilde{O}(n)$ for $(1+\eps)$-ASD. See \theoremref{thm:stream-paths} for the precise results.

\subsection{Overview of Techniques}

In this section we sketch the main ideas used in the hopsets' constructions  of \cite{C00}, in \cite{B09,N14,HKN14,HKN15}, and in our constructions.

Cohen's algorithm \cite{C00} starts with constructing a pairwise cover $\cC$ of the input graph \cite{C93,ABCP93}.   This is a collection of small-diameter clusters, with limited intersections, and such that for any path $\pi$ of length at most $W$, for a parameter $W$, all vertices of $\pi$ are clustered in the same cluster. For each cluster $C \in \cC$, the algorithm inserts into the hopset a star $\{(r_C,u) \mid u \in C\}$ connecting the center $r_C$ of $C$ with every other vertex of $C$. In addition, it adds to the hopset edges connecting centers of large clusters with one another, and recurses on small clusters.

This powerful approach has a number of limitations. First, the collection of star edges itself contains $O(\kappa \cdot n^{1+1/\kappa})$ edges, where $\kappa$ is a parameter, which controls the hopset's size. Each level of the recursion increases the exponent of the number of edges in the hopset by roughly a factor of
$\kappa \cdot n^{1/\kappa}$, and as a result, the hopset of \cite{C00} cannot be very sparse.
Second, each distance scale $[2^k,2^{k+1}]$, $k = 0,1,2,\ldots$, requires a separate hopset, and as a result, a separate collection of covers. This increases the hopset's size even further, but in addition, a hopset of scale $k+1$ in Cohen's algorithm is computed using  hopsets of all the lower scales. This results in accumulation of error, i.e., if the error incurred by each hopset computation is $1+\eps$, the approximation factor  of the ultimate hopset becomes $(1+\eps)^{\log \Lambda}$.
After rescaling $\eps' = \eps \log \Lambda$, one obtains a hopbound of roughly $(1/\eps)^\ell = O({{(\log \Lambda) \cdot \ell} \over \eps'})^\ell$, where $\ell$ is the number of levels of the recursion. As a result,  the hopbound in \cite{C00} is at least polylogarithmic in $n$.

Another line of works \cite{B09,HKN14,HKN15,N14} is based on the distance oracles and emulators\footnote{A graph $G' = (V',E',\omega')$ is called a {\em $(1+\eps,\beta)$-emulator} of an unweighted graph $G = (V,E)$, if $V \subseteq V'$, and for every pair of $u,v \in V$ of vertices, it holds that $d_{G'}(u,v) \le d_G(u,v) \le (1+ \eps) d_G(u,v) + \beta$. If $G'$ is a subgraph of $G$, then $G$ is called a {\em $(1+ \eps,\beta)$-spanner} of $G$.}  of Thorup and Zwick \cite{TZ01,TZ06}.
They build a hierarchy of sampled sets $V = A_0 \supset A_1 \supset \ldots A_{k-1} \supset A_k = \emset$, where for any $i = 1,\ldots,k-1$, each vertex $v\in A_{i-1}$ joins $A_i$ independently at random with probability $n^{-1/k}$. For each vertex $v \in V$, one can define the {\em TZ cluster} $C(v)$ by
$C(v) = \bigcup_{i=0}^{k-1} \{u \mid u \in A_i, d_G(u,v) < d_G(u,A_{i+1})\}$.  Thorup and Zwick \cite{TZ06} showed that for {\em unweighted} graphs
$H = \{(v,u) \mid u \in C(v)\}$ is a $(1+\eps,\beta)$-emulator with $O(k\cdot n^{1+1/k})$ edges, and $\beta = O(k/\eps)^k$.
Bernstein and others \cite{B09,HKN14,HKN15,N14}  showed that a closely related construction provides a hopset. Specifically, they set $k= \Theta(\sqrt{\log n})$, and build TZ clusters with respect  to $2^{\tO(\sqrt{\log n})}$-limited distances. This results in a so-called {\em restricted hopset}, i.e., a hopset $H_1$ that handles $2^{\tO(\sqrt{\log n})}$-limited distances.  Consequently, all nearly shortest paths with $N$ hops in $G$, for some $N$, translate now into
nearly shortest paths (incurring an approximation factor of $1+\eps$ of $H_1$) with ${N \over {2^{\tO(\sqrt{\log n})}}}$ hops in $G \cup H_1$. Nanongkai \cite{N14} called this operation a {\em hop reduction},
as this essentially reduces the maximum number of hops from $n-1$ to $n/2^{\sqrt{\log n}}$. Then the hop reduction is repeated for $\sqrt{\log n}$ times, until a hopset for all distances is constructed.

This scheme appears to be incapable of providing very sparse hopsets, as just the invocation of Thorup-Zwick's algorithm with $\kappa = \Theta(\sqrt{\log n})$
gives $n \cdot 2^{\Omega(\sqrt{\log n})}$ edges. In addition, the repetitive application of hop reduction blows up the hopbound to $2^{\Omega({\sqrt{\log n}})}$, i.e., the large hopbound appears to be  inherent in this approach.

Our approach combines techniques from \cite{EP04} for constructing $(1+\eps,\beta)$-spanners in unweighted graphs with those of \cite{C00},
and with a suit of new ideas.
To build their spanners,  \cite{EP04} start with constructing an Awerbuch-Peleg's partition $\cP = \{C_1,\ldots,C_q\}$ \cite{AP92} of the vertex set $V$ into disjoint clusters of small diameter. (This partition satisfies an additional property which is irrelevant to this discussion.) It then sets a distance threshold $\delta_1$ and a degree threshold $\deg_1$. Every cluster $C \in \cP$ that has at least $\deg_1$ unclustered clusters $C' \in \cP$ in its $\delta_1$-vicinity creates a supercluster which contains $C$ and these clusters. (At the beginning all clusters are unclustered. Those that join a supercluster become clustered.)
This {\em superclustering step} continues until no additional superclusters can be formed. All the remaining unclustered clusters which are at pairwise distance at most $\delta_1$ are now interconnected by shortest paths in the spanner. This is the {\em interconnection} step of the  algorithm.
Together the superclustering and interconnection steps from a single {\em phase} of the algorithm. Once the first phase is over, the same process (interleaving superclustering and interconnection) is repeated with new distance and degree thresholds $\delta_2$ and $\deg_2$, respectively, on the set of superclusters of the previous phase. The  sequences $\delta_1,\delta_2,\ldots$ and $\deg_1,\deg_2,\ldots$ are set carefully to optimize the parameters of the resulting spanner.

The basic variant of our hopset construction considers each distance scale  $[2^k,2^{k+1}]$, $k = 0,1,2,\ldots$, separately (w.l.o.g we assume all weights are at least 1). Instead of Awerbuch-Peleg's partition, we use the partition $\cP = \{\{v\} \mid v \in V\}$ into single vertices. We set the distance threshold $\delta_1$ to roughly $2^k/\beta = 2^k/(1/\eps)^\ell$, where $\ell$ is the number of phases of the algorithm, and raise it by a factor of $1/\eps$ on every phase. The degree thresholds are also set differently from the way they were set in \cite{EP04}. This is because, intuitively, the hopset contains less edges than the spanner, as the hopset can use a single edge where a spanner needs to use an entire path.
Hence the degree sequence that optimizes the hopset's size is different than the one that optimizes the spanner's size.

The superclustering and interconnection steps are also implemented in a different way than in \cite{EP04}, because of efficiency considerations. The algorithm of \cite{EP04} is not particularly efficient, and there are no known efficient streaming, distributed  or parallel implementation of it. \footnote{The algorithms of \cite{E01,EZ06} that construct $(1+\eps,\beta)$-spanners in distributed and streaming settings are not based on superclustering and interconnection technique. Rather they are based on a completely different  approach, reminiscent to that  of \cite{C00}, i.e., they build  covers, and recurse in small clusters.}
On phase $i$ we sample clusters $C \in \cP$ independently at random with probability $1/\deg_i$. The sampled clusters create superclusters of radius $\delta_i$ around them. Then the unclustered clusters of $\cP$ which are within distance $\delta_i/2$ from one another are interconnected by hopset edges.
Note that here the superclustering distance threshold and the interconnection distance thresholds differ by a factor of 2. This ensures that all involved Dijkstra explorations can be efficiently implemented. We also show that the overhead that this factor introduces to the resulting parameters of our hopset is insignificant.


Our approach (interleaving superclustering and interconnection steps) to constructing hopsets was not previously used in the hopsets' literature \cite{C00,B09, HKN14,N14,HKN15}. Rather it is adapted from \cite{EP04}. The latter paper deals with nearly-additive spanners for unweighted graphs. We believe that realizing that the technique of \cite{EP04} can be instrumental for constructing drastically improved hopsets, and adapting that technique from the context of near-additive spanners for unweighted graphs to the context of hopsets for general graphs is our main technical contribution.

To construct a hopset for {\em  all}  scales, in the centralized setting we simply take the union of the single-scale hopsets. In parallel, distributed and streaming settings, however,  Dijkstra explorations for large scales could be too expensive. To remedy this,
we rely on lower-scales hopsets for computing the current scale, like in Cohen's algorithm. On the other hand, a naive application of this approach results in polylogarithmic hopbound $\beta$. 
To achieve {\em constant} (i.e., independent of $n$) hopbound $\beta$, we compute in parallel hopsets for many different scales, using the same low-scale hopset for  distance computations.
This results in a much smaller accumulation of error than in Cohen's scheme, but requires more running time. (Roughly speaking, computing a scale-$t$ hopset using scale-$s$ hopset, for $t > s$, requires time proportional to $2^{t-s}$.)  We carefully balance this increase in running time with other parameters, to optimize the attributes of our ultimate hopset.

Finally, one needs to replace the logarithmic dependence on the aspect ratio $\Lambda$, by the same dependence on $n$. Cohen's results \cite{C00} do not have this dependence, as they rely on a PRAM reduction of Klein and Sairam \cite{KS97}, However, Klein and Sairam \cite{KS97} (see also \cite{C97} for another analysis) analyzed this reduction for single-source distance estimation, while in the hopset's case one needs to apply it to all pairs. The distributed and streaming hopsets' constructions \cite{B09,HKN14,HKN15,N14} all have a dependence on $\log \Lambda$.

We develop a new analysis of Klein-Sairam's reduction, which applies to the hopsets' scenario. We also show that the reduction can be efficiently implemented in distributed and streaming settings.


\section{Preliminaries}
\label{sec:prel}

Let $G=(V,E)$ be a weighted graph on $n$ vertices with diameter $\Lambda$, we shall assume throughout that edge-weights are positive integers.
Let $d_G$ be the shortest path metric on $G$, and let $d_G^{(t)}$ be the $t$-limited distance, that is, for $u,v\in V$, $d_G^{(t)}(u,v)$ is the minimal length of a path between $u,v$ that contains at most $t$ edges (set $d_G^{(t)}(u,v)=\infty$ if there is no such path). Note that $d_G^{(t)}$ is not a metric.

\begin{comment}
Our definition of hopsets assumes that edge weights of the input graph $G = (V,E,\omega)$ satisfy the triangle inequality, i.e., for any edge $(u,v) \in E$, we have
$d_G(u,v) = \omega(u,v)$. One can define hopsets when the input graph does not satisfy this condition, and all our results can be extended to this more general setting. However, for the sake of simplicity, we will stick to the slightly simpler setting in which triangle inequality holds.
\end{comment}

\section{Hopsets}
\label{hopset}

\subsection{A Centralized Construction}
\label{sec:cent_hop}

Let $G=(V,E)$ be a weighted graph on $n$ vertices with diameter $\Lambda$, we assume throughout the paper that the minimal distance in $G$ is 1. Fix parameters $\kappa\ge 1$, $0<\eps<1$ and $1/\kappa\le\rho<1/2$. The parameter $\beta$, which governs the number of hops our hopset guarantees, will be determined later as a function of $n,\Lambda,\kappa,\rho,\eps$.
We build separately a hopset $H_k$ for every distance range $(2^k,2^{k+1}]$, for $k \le \log \Lambda$.
We will call such a hopset $H_k$ a {\em single-scale} hopset.

Denote $\hR = 2^{k+1}$. For $\hR \le \beta = (1/\eps)^\ell$, where $\ell$ is the number of levels of the construction (to be determined), an empty hopset $H_k =\emptyset$ does the job. Hence we assume that $k > \log \beta -1$, i.e., $\hR > \beta$.

The algorithm initializes the hopset $H_k$ as an empty set, and proceeds in phases. It starts with setting $\chP_0 = \{\{v\} \mid v \in V\}$ to be the partition of $V$ into singleton clusters.
The partition $\chP_0$ is the input of phase 0 of our algorithm. More generally, $\chP_i$ is the input of phase $i$, for every index $i$ in a certain appropriate range, which we will specify in the sequel.

Throughout the algorithm, all clusters $C$ that we will construct will be centered at designated centers $r_C$. In particular, each singleton cluster $C= \{v\} \in \chP_0$ is centered at $v$. We define $\Rad(C) = \max \{d_{G(C)}(r_C,v) \mid v \in C\}$, and $\Rad(\chP_i) = \max_{C \in \chP_i} \{\Rad(C)\}$.

All phases of our algorithm except for the last one consist of two steps. Specifically, these are the {\em superclustering} and the {\em interconnection} steps.
The last phase contains only the interconnection step, and the superclustering step is skipped.
We also partition the phases
into two {\em stages}. The first stage consists of phases  $0,1, \ldots,i_0=\lfloor\log(\kappa\rho)\rfloor$, and the second stage consists of all the other phases $i_0+1,\ldots,i_1$ where $i_1=i_0+\left\lceil\frac{\kappa+1}{\kappa\rho}\right\rceil-2$, except for the last phase $\ell = i_1+1$. The last phase will be referred to as the {\em concluding phase}. 

Each phase $i$ accepts as input two parameters, the distance threshold parameter $\delta_i$, which determines the range of the Dijkstra explorations, and the degree parameter $\deg_i$, which determines the sampling probability. The difference between stage 1 and 2 is that in stage 1 the degree parameter grows exponentially, while in stage 2 it is fixed. The distance threshold parameter
grows in the same steady rate (increases by a factor of $1/\eps$) all through the algorithm.

The distance thresholds' sequence is given by $\alpha = \eps^\ell \cdot \hR$, $\delta_i = \alpha (1/\eps)^i + 4 R_i$, where $R_0 = 0$ and $R_{i+1} = \delta_i + R_i = \alpha (1/\eps)^i + 5 R_i$, for $i \ge 0$. It follows that $R_1 = \alpha$, and by estimating the recurrence we obtain $R_i \le 2 \cdot \alpha \cdot (1/\eps)^{i-1}$.
The degree sequence in the first stage of the algorithm is given by $\deg_i = n^{{2^i}/ \kappa}$, for $i = 0,1,\ldots,i_0$. We then use $\deg_i = n^\rho$ in all subsequent phases $i_0+1,\ldots,i_1$. Finally, on phase $\ell = i_1 + 1$ we perform just the interconnection step. Note that $\ell\ge 2$ since $\rho<1/2$.


Next we take a closer look on the execution of phase $i$, $i = 0,1,2,\ldots,\ell-1$. At the beginning of the phase we have a collection $\chP_i$ of clusters, of
radius $2\alpha \cdot (1/\eps)^{i-1}$, for $i \ge 1$, and radius 0 for $i = 0$. (It will be shown in Claim \ref{claim:conn-cluster}
that $\Rad(\chP_i) \le R_i = 2\alpha \cdot (1/\eps)^{i-1}$, for all $i = 0,1,\ldots,\ell$.)
 Each of these clusters is now sampled with probability $1/\deg_i$, i.a.r.. The resulting set of sampled clusters is denoted $\cS_i$.
We then initiate a single Dijkstra exploration in $G$ rooted at the set $\Rt = \{r_C \mid C \in \cS_i\}$ of cluster centers of sampled clusters. The Dijkstra exploration is conducted to depth $\delta_i$. Let $F_i$ denote the resulting forest.

Let $C' \in \chP_i \setmns \cS_i$ be a cluster whose center $r_{C'}$ was reached by the exploration, and let $r_C$, for some cluster $C \in \cS_i$, be the cluster center such that $r_{C'}$ belongs to the tree of $F_i$ rooted at $r_C$. We then add an edge $(r_C,r_{C'})$ of weight $\omega(r_C,r_{C'}) = d_G(r_C,r_{C'})$ into the hopset $H_k$, which we are now constructing.
A supercluster $\hC$ rooted at $r_{\hC} = r_C$ is now created. It contains all vertices of $C$ and of clusters $C'$ as above. This completes the description of the superclustering step. The resulting set $\chS_i$ of superclusters becomes the next level partition $\chP_{i+1}$, i.e., we set $\chP_{i+1} \leftarrow \chS_i$.

\begin{claim}\label{claim:conn-cluster}
Fix any cluster $C\in\chP_i$ with center $r_C$. Then for any $u\in C$ there is a path in $H_k$ of at most $i$ edges from $r_C$ to $u$ of length at most $R_i$.
\end{claim}
\proof
The proof is by induction on $i$, the basis $i=0$ holds as $C$ is a singleton. Assume it holds for $i$, and fix any $\hC\in\chP_{i+1}$ and $u\in \hC$. Recall that $\hC$ consists of a sampled cluster $C\in \cS_i$, and clusters $C'\in \chP_i$ for which the Dijkstra exploration to range $\delta_i$ from $r_C$ reached their center $r_{C'}$. Assume $u\in C'$ (the case where $u\in C$ is simpler). Then by induction there is a path of length at most $R_i$ from $r_{C'}$ to $u$ in $H_k$ of $i$ hops, and by construction we added the edge $(r_C,r_{C'})$ of weight $d_G(r_C,r_{C'})$ into the hopset $H_k$. This implies a path of $i+1$ hops and length at most
\[
\delta_i+R_i=R_{i+1}~.
\]
\QED


Let $\chU_i$ denote the set of $\chP_i$ clusters which were not superclustered into $\chS_i$ clusters. These clusters are involved in the interconnection step. Specifically, each of the cluster centers $r_C$, $C \in \chU_i$, initiates now a separate Dijkstra exploration to depth $\half \delta_i = \half \alpha \cdot (1/\eps)^i + 2R_i$. For any cluster center $r_{C'}$ of a cluster $C' \in \chU_i$ such that $r_{C'}$ was discovered by an exploration originated at $r_C$, we now insert an edge $(r_C,r_{C'})$ into the hopset, and assign it weight $\omega(r_C,r_{C'}) = d_G(r_C,r_{C'})$.  This completes the description of the interconnection step.

\begin{lemma}
\label{lm:explorations}
For any vertex $v \in V$, the expected number of explorations that visit $v$ at the interconnection step of phase $0\le i\le i_1$ is at most $\deg_i$.
\end{lemma}
\proof
For $0\le i\le i_1$, assume that there are $l$ clusters of $\chP_i$ within distance $\delta_i/2$ from $v$.  If at least one of them is sampled to $\cS_i$, then no exploration will visit $v$ (since in the superclustering phase the sampled center will explore to distance $\delta_i$, and thus all these $l$ cluster will be superclustered into some cluster of $\chS_i$). The probability that none of them is sampled is $(1-1/\deg_i)^l$, in which case we get that $l$ explorations visit $v$, so the expectation is $l\cdot(1-1/\deg_i)^l\le \deg_i$ for any $l$.
\QED

A similar argument yields the following Lemma.
\begin{lemma}
\label{lm:property}
For any constant $c > 1$, with probability at least $1 - 1/n^{c-1}$, for every vertex $v \in V$, at least one among the $\deg_i \cdot c \cdot \ln n$ closest cluster centers $r_{C'}$ with $C' \in \chP_i$ to $v$ is sampled, i.e., satisfies
$C' \in \cS_i$.
\end{lemma}

We analyze the number of clusters in collections $\chP_i$ in the following lemma.

\begin{lemma}
\label{lm:Pi}
Assuming $n^{\rho}=\omega(1)$, with high probability, for every $i = 0,1,\ldots,i_0+1$ we have
\begin{equation}
\label{eq:Pi}
|\chP_i| ~\le~ 2 \cdot n^{1 - {{2^i -1} \over \kappa}} ~,
\end{equation}
and for $i=i_0+2,\ldots,i_1+1$,
$$|\chP_{i}| ~\le~ 2 \cdot n^{1 + 1/\kappa - (i - i_0)\rho} ~.$$
\end{lemma}
\proof
For the first assertion, the probability that a vertex $v\in V$ will be a center of a cluster in $\chP_i$ is $\prod_{j=0}^{i-1}1/\deg_j=n^{-(2^i-1)/\kappa}$. Thus the expected size of $\chP_i$ is $n^{1-(2^i-1)/\kappa}$, and as these choices are made independently, by Chernoff bound,
\[
\Prob[|\chP_i|\ge 2\Expect[|\chP_i|]]\le \exp\{-\Omega(\Expect[|\chP_i|])\}=\exp\{-\Omega(n^{1 - {{2^i -1} \over \kappa}})\}~.
\]

Since for $\rho < 1/2$ and $i \le i_0+1 = \lfloor \log \rho \kappa \rfloor+1$, we have
$n^{1 - {{2^i - 1} \over \kappa}}  \ge n^{1-2\rho} = \omega(\log n)$, we conclude that whp for all $0\le i\le i_0+1$,
$|\chP_{i}| \le 2 n^{1 - {{2^{i} - 1} \over \kappa}}$. In particular, $|\chP_{i_0+1}| = O( n^{1 - \rho +1/\kappa})$.

For the second assertion, consider any $i \in [i_0+2,i_1+1]$, the expected size of $\chP_{i}$ is
\[
\Expect[|\chP_{i}|]=n\cdot\prod_{j=0}^{i-1}1/\deg_j\le n^{1 + 1/\kappa -\rho - (i -1 - i_0) \rho}=n^{1 + 1/\kappa - (i- i_0) \rho}~.
\]
Since $n^{1 + 1/\kappa - (i- i_0) \rho}\ge n^\rho$ for any $i\le i_1$, by Chernoff bound with probability at least $1 - \exp\{-\Omega(n^\rho)\}$ (which is $1-o(1)$ by our assumption on $n^{\rho}$), we have
$$|\chP_{i}| ~\le~ 2 \cdot n^{1 + 1/\kappa - (i - i_0)\rho} ~.$$
\QED
This lemma implies that whp
\begin{equation}
\label{eq:finalPi}
|\chP_{i_1 + 1}| ~\le~ O(n^{1 + 1/\kappa  - (i_1 + 1 - i_0) \rho})
~= ~ O(n^{1 +1/\kappa - (\lceil {{\kappa + 1} \over {\kappa\rho}} \rceil - 1)\rho} ) ~=~ O(n^\rho)~.
\end{equation}
For the assumption of the Lemma above to hold, we will need to assume that $\rho \ge {{\log\log n} \over {2\log n}}$, say. We will show soon that this assumption is valid in our setting.

The running time required to implement the single Dijkstra exploration in the superclustering of phase $i$ is $O(|E|+n\log n)$, while in the interconnection step, by \lemmaref{lm:explorations} every vertex is expected to be visited by at most $\deg_i$ explorations, so the expected running time of phase $0\le i\le i_1$ is $O(|E|+n\log n)\cdot\deg_i$. Recall that in the last phase $i_1+1$ there is no superclustering step, but as \eqref{eq:finalPi} implies, there are whp only $O(n^\rho)$ clusters, so each vertex will be visited at most $O(n^\rho)$ times. Thus the total expected running time is
\begin{eqnarray*}
O(|E|+n\log n)\cdot\left( \sum_{i=0}^{\ell-1} (\deg_i) +n^\rho\right)&=& O(|E|+n\log n)  \cdot \left(\sum_{i=0}^{i_0}(n^{2^i/\kappa})+(i_1-i_0)n^{\rho}\right)\\
&=&  O(|E|+n\log n) \cdot (n^{2^{i_0}/\kappa}+n^\rho/\rho)\\
&=& O(|E|+n\log n) \cdot n^\rho/\rho~.
\end{eqnarray*}



The size of the hopset $H_k$ that was constructed by this algorithm is dominated by the number of edges inserted by the interconnection steps, since all the edges inserted at superclustering steps induce a forest.
Due to \lemmaref{lm:explorations}, the expected number of edges inserted by the interconnection step of phase $i$ is at most $O(|\chP_i| \cdot \deg_i) = O(n^{1+1/\kappa})$, for $i \le i_0$, and
$\sum_{i=i_0+1}^{\ell +1} O(|\chP_i| \cdot \deg_i) = O(n^{1 + 1/\kappa})$ edges on the later phases.
Hence overall $\Expect(|H_k|) = O(n^{1+1/\kappa} \cdot \log \kappa)$. We remark that the factor $\log \kappa$ can be eliminated from the hopset size by using a refined degree sequence, at the cost of increasing the number of phases by 1 (this will increase the exponent of $\beta$ by 1). We elaborate on this at \sectionref{sec:less-edges}. Then the number of edges contributed to the hopset $H_k$ by all interconnection steps becomes $O(n^{1+1/\kappa})$.

Next we analyze  the stretch and the hopbound of $H_k$. Write $H = H_k$.
Observe that, by Claim \ref{claim:conn-cluster}, $\Rad(\chU_0) = R_0 = 0$, and, for all $i \in [1,\ell]$, $\Rad(\chU_i) \le \Rad(\chP_i)  \le R_i \le 2\alpha (1/\eps)^{i-1}$. (We assume $\eps < 1/10$, and later justify this assumption.)
Write $c = 2$.
Note also that for any pair of distinct clusters $C,C' \in \chU_i$, for any $i$, which are at distance $d_G(C,C') \le \half \alpha \cdot (1/\eps)^i$,
it holds that $d_G(r_C,r_{C'}) \le d_G(C,C') + 2R_i \le \half \alpha (1/\eps)^i + 2 \cdot R_i= \half \delta_i$. Hence for every pair of clusters $C,C'$ as above, an edge $(r_C,r_{C'})$ of weight $\omega(r_C,r_{C'}) = d_G(r_C,r_{C'})$ belongs to the hopset.

Observe that $\chU = \bigcup_{i=0}^\ell \chU_i$  is a partition of $G$.  For any $i$, we denote $\chU^{(i)} = \bigcup_{j=0}^i \chU_j$.

\begin{lemma}
\label{lm:hop_str}
Let $x,y$ be a pair of vertices with $d_G(x,y) \le \half \alpha \cdot (1/\eps)^i$ and such that all vertices of a shortest path $\pi(x,y)$ in $G$ between them are clustered in  $\chU^{(i)}$, for some $i \le \ell$.
Then it holds that
\begin{equation}
\label{eq:hop_str}
d_{G \cup H}^{(h_i)}(x,y) \le d_G(x,y)( 1 +16 c(i-1) \cdot \eps) + 8 \cdot \alpha \cdot c \cdot (1/\eps)^{i-1}~,
\end{equation}
with $h_i$ given by $h_0 = 1$, and $h_{i+1} = (h_i + 1)(1/\eps+2)+2i+5$.
\end{lemma}
\proof
The proof is by induction on $i$. The basis is the case $i = 0$.
\inline Basis: We assume $d_G(x,y) \le \half \alpha$ and all vertices of $\pi(x,y)$ are clustered in $\chU_0$, then there is an edge $(x,y)$ in $H$ with $\omega(x,y) = d_G(x,y)$, and indeed
\[
d_{G \cup H}^{(h_0)}(x,y)=d_G(x,y)\le d_G(x,y)( 1 -16 c\cdot \eps) + 8 \cdot \alpha \cdot c \cdot (1/\eps)^{-1}~.
\]
\inline  Step:  We assume the assertion of the lemma for some index $i$, and prove it for $i+1$. \\
Consider first a pair $u,v$ of vertices such that all vertices of $\pi(u,v)$ are clustered in $\chU^{(i)}$, for a fixed $i < \ell$, without any restriction on $d_G(u,v)$.

We partition $\pi(u,v)$ into segments $L_1,L_2,\ldots$ of length roughly $\half \cdot \alpha \cdot (1/\eps)^i$ each in the following way. The first segment $L_1$ starts at $u$, i.e., we write $u = u_1$. Given a left endpoint $u_p$, $p \ge 1$, of a segment $L_p$, we set the right endpoint $v_p$ of $L_p$ to be (if exists) the farthest vertex of $\pi(u,v)$ from $u_p$ which is closer to $v$ than $u_p$, and such that $d_G(u_p,v) \le \half \cdot \alpha \cdot (1/\eps)^i$.

If $v_p$ does not exist then the $p$th segment $L_p$ is declared as {\em void}, and we define $v_p = u_{p+1}$ to be the neighbor of $u_p$ on $\pi(u,v)$ which is closer to $v$. If $v_p$ does exist, then $u_{p+1}$ is (if exists) the "right" neighbor of $v_p$ on $\pi(u,v)$, i.e., the neighbor of $v_p$ which is closer to $v$ than $v_p$ is.
(It may not exist only if $v_p = v$.) Observe that in either case, if $u_{p+1}$ exists then $d_G(u_p,u_{p+1}) > \half \cdot \alpha \cdot (1/\eps)^i$.

We also define {\em extended} segments $\hL_p$ in the following way. If $L_p$ is a void segment, then we define $\hL_p = L_p$.  Otherwise $\hL_p$ is the segment of $\pi(u,v)$ connecting $u_p$ with $u_{p+1}$, if $u_{p+1}$ exists, and with $v_p$ otherwise. (This may be the case only if $L_p = \hL_p$ is the last, i.e., the rightmost, segment of the path $\pi(u,v)$.)

Observe that every non-void extended segment $\hL_p$, except maybe the last one, has length at least $\half \cdot \alpha \cdot (1/\eps)^i$, and every segment $L_p$ has length at most $\half \cdot \alpha \cdot (1/\eps)^i$.

Next we construct a path $\pi'(u,v)$ in $G \cup H$, which has roughly the same length as $\pi(u,v)$, but consists of much fewer hops.
Consider a segment $L_p$, with left endpoint $u_p$ and right endpoint $v_p$, and its extended segment $\hL_p$ with right endpoint $u_{p+1}$.
We define a {\em substitute} segment $L'_p$ in $G \cup H$, connecting $u_p$ with $u_{p+1}$ with a few hops, and of roughly the same length.

If $L_p$ is a void segment then $L'_p$ is just the single edge $(u_p,u_{p+1})$, taken from $E = E(G)$.
Observe that for a void segment,
$$\omega(\hL_p) = \omega(L_p) = \omega(u_p,u_{p+1}) = d_G(u_p,u_{p+1})~.$$
Otherwise, if $L_p$ is not a void segment, then $d_G(u_p,v_p) \le \half \cdot \alpha \cdot (1/\eps)^i$.
Observe also that since all vertices of $\pi(u,v)$ are $\chU^{(i)}$-clustered, this is also the case for the subpath $\pi(u_p,v_p)$. Hence the induction hypothesis is applicable to this subpath, and so there exists a path $\pi'(u_p,v_p)$ in $G \cup H$ with at most $h_i$ hops, such that
$$\omega(\pi'(u_p,v_p)) ~\le~ d_G(u_p,v_p) \cdot (1 + 16 \cdot c(i-1) \cdot \eps) + 8 \cdot \alpha \cdot c \cdot (1/\eps)^{i-1}~.$$
We define $L'_p$ to be the concatenation of $\pi'(u_p,v_p)$ with the edge $(v_p,u_{p+1})$. (This edge is taken from $G$.)
Since $v_p$ lies on a shortest path between $u_p$ and $u_{p+1}$, it follows that
$$\omega(L'_p) ~\le~ (1 +16 c \cdot (i-1) \eps) \cdot d_G(u_p,u_{p+1}) + 8 \cdot \alpha \cdot c \cdot (1/\eps)^{i-1}~,$$
and $L'_p$ contains up to $h_i + 1$ hops.

Finally, our ultimate path $\pi'(u,v)$ is the concatenation of all the substitute segments $L'_1 \circ L'_2 \circ \ldots \circ L'_q$, where $\pi(u,v) = \hL_1 \circ \hL_2 \circ \ldots \circ \hL_q$.
%
Since each extended segment  has length at least $\half \alpha \cdot (1/\eps)^{i-1}$
 we conclude that
\begin{eqnarray}
\nonumber
d_{G\cup H}^{((h_i+1) \cdot \lceil  {{d_G(u,v)} \over {\half \alpha (1/\eps)^i}} \rceil)} (u,v) & \le &
d_G(u,v)(1  +16c (i-1)\eps) +  \lceil {{d_G(u,v)} \over {\half \alpha (1/\eps)^i}} \rceil \cdot 8 \alpha  c \cdot (1/\eps)^{i-1}
\\
\nonumber
& \le & d_G(u,v)\left(1 + 16 c(i-1) \eps + {{8 \cdot \alpha \cdot c (1/\eps)^{i-1}} \over {\half \alpha \cdot (1/\eps)^i}}\right) + 8  \alpha c \cdot (1/\eps)^{i-1} \\
\label{eq:short_str}
& = & d_G(u,v)(1 + 16c \cdot i \cdot \eps) + 8 \cdot \alpha c \cdot (1/\eps)^{i-1}~.
\end{eqnarray}
Now consider $x,y$ such that $d_G(x,y) \le \half \alpha \cdot (1/\eps)^{i+1}$ and such that $\pi(x,y)$ is  $\chU^{(i+1)}$-clustered. Let $z_1$ and $z_2$ denote the leftmost and the rightmost $\chU_{i+1}$-clustered vertices on this path, and denote by $C_1$ and $C_2$ their respective clusters. Denote also $r_1 = r_{C_1}, r_2 = r_{C_2}$.
Denote also by $w_1$ (resp., $w_2$) the neighbor of $z_1$ (resp., $z_2$) on the subpath $\pi(x,z_1)$ (resp., $\pi(z_2,y)$) of $\pi(x,y)$.

The path $\pi'(x,y)$ in $G \cup H$ between $x$ and $y$ is constructed in the following way. By (\ref{eq:short_str}), we can reach from $x$ to $w_1$ while incurring a multiplicative stretch of $(1 + 16 c i \cdot \eps)$ and an additive error of $8 \cdot \alpha \cdot c \cdot (1/\eps)^{i-1}$, and using at most $b_1=(h_i +1) \cdot \lceil  {d_G(x,w_1) \over {\half \alpha (1/\eps)^i}} \rceil$ hops.
The same is true for the pair $w_2,y$, except that the required number of hops is at most   $b_2=(h_i +1) \cdot \lceil  {d_G(w_2,y) \over {\half \alpha (1/\eps)^i}} \rceil$.
Finally, the path $\pi'(x,y)$ connects $w_1$ with $w_2$ via edges $(w_1,z_1),(z_2,w_2)$ that belong to $E(G)$, the edge $(r_1,r_2)$ of the hopset $H$, and the paths $\pi(z_1,r_1),\pi(r_2,z_2)$ (each of $i+1$ hops) in $H$ given by \claimref{claim:conn-cluster}. Hence
\begin{eqnarray*}
d^{(h_{i+1})}_{G \cup H}(x,y)  & \le & d^{(b_1)}_{G \cup H}(x,w_1) + d^{(1)}_G(w_1,z_1) + d^{(i+1)}_G(z_1,r_1) + d^{(1)}_H(r_1,r_2)\\
&& + ~d^{(i+1)}_G(r_2,z_2) + d^{(1)}_G(z_2,w_2) + d^{(b_2)}_{G \cup H}(w_2,y) \\
& \le & (1 + 16c \cdot  i \cdot \eps) d_G(x,w_1) + d_G(w_1,z_1) + R_{i+1} + (d_G(z_1,z_2) + 2R_{i+1})
\\
&&+ ~R_{i+1} + d_G(z_2,w_2) + (1+ 16ci \cdot \eps) d_G(w_2,y) + 2 \cdot(8 \alpha  c \cdot (1/\eps)^{i-1}) \\
& \le & (1 +16 c \cdot i \cdot \eps)d_G(x,y) + 4 \cdot \alpha c (1/\eps)^i + 16\alpha c \cdot (1/\eps)^{i-1} \\
& \le &(1 + 16 c \cdot  i\cdot \eps)d_G(x,y) + 8 \cdot \alpha c \cdot (1/\eps)^i~,
\end{eqnarray*}
where the required number of hops indeed satisfies
\begin{eqnarray}
\label{eq:hop_bnd}
\lefteqn{(h_i+1) \left(\lceil {{d_G(x,w_1)} \over {\half \alpha \cdot (1/\eps)^i}} \rceil + \lceil {{d_G(w_2,y)} \over {\half \alpha \cdot (1/\eps)^i}} \rceil\right) +2i+ 5}\nonumber\\
\nonumber
& \le & (h_i +1)\left({{d_G(x,y)} \over {\half \alpha (1/\eps)^i}} + 2\right) +2i+ 5\\\nonumber
& \le & (h_i + 1)(1/\eps+2) +2i+ 5\\
&=&h_{i+1}~.
\end{eqnarray}
\QED

The recursive equation $h_{i+1} = (h_i+1) (1/\eps + 2) +2i+ 5$ solves to $h_i  \le 3 \cdot (1/\eps+2)^i$, for $\eps < 1/10$, i.e., $h_\ell \le 3 \cdot (1/\eps+2)^\ell$. Write $\zeta = 16c (\ell+1)\cdot \eps$ and $\beta=2h_\ell+1\le 6 \cdot (1/\eps+2)^\ell+1$.

\begin{corollary}
Let $x,y\in V$ be such that $d_G(x,y)\in(\hat{R}/2,\hat{R}]$. Then
\[
d_{G\cup H}^{(\beta)}(x,y)\le (1+\zeta)\cdot d_G(x,y)~.
\]
\end{corollary}
\proof
Let $\pi(x,y)$ be the shortest path in $G$ between $x,y$. By a similar (and simpler) argument to the one appearing in \lemmaref{lm:hop_str}, one can see that there exists an edge $(u,v)\in E$ such that both $u,v$ are on $\pi(x,y)$, and also $d_G(x,u)\le\hat{R}/2$ and $d_G(y,v)\le\hat{R}/2$. Applying \lemmaref{lm:hop_str} on these pairs with $i=\ell$, recalling that $\half \alpha \cdot (1/\eps)^\ell=\hat{R}/2$ and that every vertex is clustered in $\chU^{(\ell)}$, it follows that
\[
d_{G\cup H}^{(h_\ell)}(x,u)\le d_G(x,u)(1 + 16 c (\ell-1) \cdot \eps) +8c\cdot\eps\cdot\hat{R}~.
\]
Similarly also
\[
d_{G\cup H}^{(h_\ell)}(y,v)\le d_G(x,u)(1 + 16 c (\ell-1) \cdot \eps) +8c\cdot\eps\cdot\hat{R}~.
\]
Since $\beta=2h_\ell+1$ and $(u,v)\in E$, we obtain
\begin{eqnarray*}
d_{G\cup H}^{(\beta)}(x,y)&\le&d_{G\cup H}^{(h_\ell)}(x,u)+d^{(1)}_G(u,v)+d_{G\cup H}^{(h_\ell)}(y,v)\\
&\le&(d_G(x,u)+d_G(u,v)+d_G(v,y))\cdot(1 + 16 c (\ell-1) \cdot \eps)+16c\cdot\eps\cdot\hat{R}\\
&\le&d_G(x,y)\cdot(1 + 16 c (\ell-1) \cdot \eps)+32c\cdot\eps\cdot d_G(x,y)\\
&=&d_G(x,y)\cdot(1+\zeta)~.
\end{eqnarray*}
\QED

Recall that $\ell = \lfloor\log (\kappa\rho)\rfloor + \lceil {{\kappa +1 } \over {\rho\kappa}} \rceil - 1\le \log (\kappa\rho)+ \lceil 1/\rho \rceil $ is the number of phases of the algorithm (for the sake of brevity, from now on we shall ignore the ceiling of $1/\rho$). When we rescale $\eps=\zeta$ as the strech factor then $\beta=O(\ell/\eps)^\ell=O\left( {{\log \kappa + 1/\rho} \over \eps} \right)^{\log \kappa + 1/\rho}$.

Our ultimate hopset $H$ is created by $H \leftarrow \bigcup _{k > \log \beta-1} H_k$, i.e., $H$ is the union of up to $\lceil \log \Lambda \rceil$ hopsets, each of which takes care of its own distance range. As a result, the number of edges in $H$ is  $O(n^{1+1/\kappa} \cdot \log \Lambda)$, and its expected construction time is  $O((|E| + n \log n) \cdot n^\rho/\rho \cdot \log \Lambda )$. The following theorem summarizes this result.
\begin{theorem}
\label{thm:hop_exist}
For any graph $G = (V,E)$ with $n$ vertices and diameter $\Lambda$, $2\le\kappa\le (\log n)/4$, $1/2 > \rho \ge 1/\kappa$, and $0 < \eps \le 1$,
our algorithm constructs a $(\beta,\eps)$-hopset $H$ with $O(n^{1+1/\kappa} \cdot \log \Lambda)$ edges in expectation,
in time $O((|E| + n \log n) (n^\rho/\rho\cdot \log \Lambda) )$, with
$\beta =O\left( {{\log \kappa + 1/\rho} \over \eps} \right)^{\log \kappa + 1/\rho}$.

Moreover, the hopset consists of up to $\lceil \log \Lambda \rceil$ single-scale hopsets. Each of these hopsets $H_k$ has the same $\beta$, and its expected size is  $|H_k| = O(n^{1+1/\kappa})$. It can be constructed in $O((|E| + n \cdot \log n) (n^\rho/\rho))$ time.

\end{theorem}

In \sectionref{sec:reduction} we will show how to remove the dependence on the aspect ratio $\Lambda$, and replace it with $n$, which yields the following.

\begin{theorem}
\label{thm:hop1}
For any graph $G = (V,E,\omega)$ with $n$ vertices, $2\le\kappa\le (\log n)/4$, $1/2 > \rho \ge 1/\kappa$, and $0 < \eps \le 1$,
our algorithm constructs a $(\beta,\eps)$-hopset $H$ with $O(n^{1+1/\kappa} \cdot \log n)$ edges in expectation,
in time $O((|E| + n \log n) (n^\rho/\rho\cdot \log n) )$, with
$\beta =O\left( {{\log \kappa + 1/\rho} \over \eps} \right)^{\log \kappa + 1/\rho}$.
\end{theorem}

Finally, we note that our assumption that $\rho>\log\log n/(2\log n)$ is justified, as otherwise we get $\beta\ge n$, in which case an empty hopset will do. Also $\eps\le 1$, because we rescaled it by a factor
of $ 16c (\ell+1)> 10$.

\subsubsection{Improved Hopset Size}
\label{sec:less-edges}

Here we show how select a refined degree sequence that will eliminate the term of $\log\kappa$ from the size of the hopset constructed at \sectionref{sec:cent_hop}, by increasing the number of phases by 1 (and thus the exponent of $\beta$ by an additive 1).
Specifically, one can  set $\deg_i = n^{2^i/ \kappa}/2^{2^i-1}$, for each $i = 0,1,\ldots,i_0=\lfloor\log(\kappa\rho)\rfloor$.
As a result we get that $\Expect[|\chP_i|]=n\cdot\prod_{j=0}^{i-1}1/\deg_j=n^{1-\frac{2^i-1}{\kappa}}\cdot 2^{2^i-1-i}$, and thus the expected number of edges inserted at phase $i\le i_0$ is at most
\[
O(\Expect[|\chP_i|]\cdot\deg_i)=O(n^{1+1/\kappa}/2^i)~,
\]
and thus it is $O(n^{1+1/\kappa})$ over all phases $i=0,1,\ldots,i_0$.
When the first stage concludes, we run the phase $i_0+1$ with $\deg_{i_0+1} = n^{\rho/2}$, and all subsequent phases with $\deg_i = n^\rho$. To bound the expected number of edges added at phase $i_0+1$ we need to note that $2^{2^{i_0+1}}\le 2^{2\kappa\rho}\le n^{\rho/2}$ as long as $\kappa\le (\log n)/4$. (The latter can be  assumed without affecting any of the parameters by more than a constant factor). It follows that $\Expect[|\chP_{i_0+1}|]\cdot\deg_{i_0+1}= n^{1-\frac{2^{i_0+1}-1}{\kappa}}\cdot 2^{2^{i_0+1}-1-(i_0+1)}\cdot n^{\rho/2}\le n^{1+1/\kappa}$. In the remaining phases $\Expect[|\chP_{i_0+i}|]\le n^{1+1/\kappa-(i-1)\rho}$ for $i\ge 2$, and the contribution of these phases is a converging sequence. In particular, at phase $i_0+2$ we have $\Expect[|\chP_{i_0+2}|]\le n^{1+1/\kappa-\rho}$. We can take $i_1= i_0+ \lceil\frac{\kappa+1}{\kappa\rho}\rceil-1$, and finally $\Expect[|\chP_{i_1+1}|]\le n^{1+1/\kappa-(i_1+1-i_0-1)\rho}\le n^{\rho}$, as required to bound the running time of the final interconnection phase. The total number of phases is now $\ell=i_1+1=\lfloor\log(\kappa\rho)\rfloor+\lceil\frac{\kappa+1}{\kappa\rho}\rceil$, which might be larger by an additive 1 than the bound claimed in \theoremref{thm:hop1}. For the sake of brevity, we shall ignore this small overhead.

\subsection{Distributed Implementation in Congested Clique Model}
\label{sec:hopset_clique}

In this section we argue that our hopset construction admits an efficient implementation in the distributed Congested Clique model, albeit with somewhat worse parameters.

A direct implementation of the algorithm from the \sectionref{sec:cent_hop} in this model requires up to $O(n)$ distributed time, because a Dijkstra algorithm invoked during the hopset's construction may explore paths with up to $n$ hops.
To overcome this issue we use the following idea, which dates back to Cohen's work \cite{C00}. Specifically, we use hopsets $\bigcup_{\log \beta -1 < j \le k-1} H_j$ to construct the hopset $H_k$.
Indeed, the hopset $H_k$ needs to take care of vertex pairs $u,v$ with $d_G(u,v) \in (2^k,2^{k+1}]$, $\hR^{(k)} = 2^{k+1}$, while $E \cup H^{(k-1)}$ (with $H^{(k-1)} = \bigcup_{j \le k-1} H_j$) provides $(1+\eps)$-approximate shortest paths with up to $\beta $ hops for pairs $u,v$ with $d_G(u,v) \le 2^k$.
Denote $E'^{(k-1)} = E \cup H^{(k-1)}$.

Consider a vertex pair $u,v \in V$, with $d_G(u,v) \in (2^k,2^{k+1}]$, and let $\pi(u,v)$ be a shortest path between them. Let $x \in V(\pi(u,v))$ be the farthest vertex of $\pi(u,v)$ from $u$ which is at distance no greater than $2^k$ from $u$, and let $y$ be its neighbor such that $d_G(u,y) > 2^k$.
Observe that the hopset $H^{(k-1)}$ provides (together with the edge set $E$ of the original graph) a $(1+\eps)$-approximate shortest $u-v$ path with at most $2\beta + 1$ hops.
This  path $\pi'(u,v)$ is a concatenation of a $(1+\eps)$-approximate $u-x$ path $\pi'(u,x)$ in $E'^{(k-1)}$ with the edge $(x,y)$ and with a $(1+\eps)$-approximate $y-v$ path $\pi'(y,v)$ in
 $E'^{(k-1)}$.

Next we generalize this observation.

\begin{lemma}
\label{lm:aux_hopset}
For any index $p$, $\log \beta - 1 \le p \le k-1$, the hopset $H^{(p-1)}$ provides (together with edges of $E$) $(1+ \eps)$-approximate shortest paths with at most $2^{k+1-p}(\beta+1)$ hops for pairs $u$,$v$, with $d_G(u,v) \in (2^k,2^{k+1}]$,
\end{lemma}
\proof
Partition $\pi(u,v)$ into segments of length at most $2^{p-1}$, except for void segments. (See the proof of Lemma \ref{lm:hop_str} for the definition of void, extended and substitute segments.)
For each non-void extended segment we use at most $\beta + 1$ hops of the substitute segment in $G \cup H^{(p-1)}$, and for each void extended segment  we use just one hop. Every extended segment has length at least $2^{p-1}$, and thus $\pi(u,v)$ is partitioned into at most $2^{k+1-p}$ such segments.
Hence the assertion of the lemma follows. (In fact, in the last segment $\beta$ hops suffice, and thus the estimate can be refined to $2^{k+1-p}(\beta+1) - 1$.)
\QED

The first variant of our distributed algorithm builds hopsets $H_{k_0}, H_{k_0+1},H_{k_0+2},\ldots$, $k_0 = \lceil \log \beta \rceil$, one after another. Suppose that all hopsets $H_j$, $j < k$, were already built, and we are now building $H_k$. We only need to describe how the superclustering and the interconnection steps on a phase $i$, $0 \le i \le \ell$ are implemented. (On phase $\ell$ there is no superclustering step, and there we only need to implement the interconnection step.)

Denote $\zeta =\zeta_{k-1}$ the value such that $E'^{(k-1)}$ guarantees stretch $1+\zeta$.
We slightly modify the sequence of distance thresholds $\delta_i$, specifically, we multiply them all by a factor of $1+\zeta$.

We define $R'_i = (1+\zeta)R_i$ and $\delta'_i=(1+\zeta)\delta_i$, for every $i \in [0,\ell]$. Here $\alpha = \alpha^{(k)} = {{\hR^{(k)}} \over {(1/\eps)^\ell}}$, where $\hR^{(k)}=2^{k+1}$.
The distributed variant of our algorithm uses distance thresholds $(\delta'_i \mid 0 \le i \le \ell)$, and as a result upper bounds on the radii of $i$-level clusters become $R'_i = R_i(1+\zeta)$.

The superclustering step of our {\em centralized} algorithm conducts a Dijkstra exploration from the set $\Rt = \{r_c \mid C \in \cS_i\}$ in $G$ to distance $\delta_i = \alpha \cdot (1/\eps)^i + 4 R_i \le \hR^{(k)}$.
Instead in the distributed version we invoke  Bellman-Ford algorithm originated in $\Rt$ over the edge set $E'^{(k-1)} = E \cup H^{(k-1)}$ for $2\beta+1$ rounds.

Specifically, vertices $r_C \in \Rt$ initialize  their distance estimates $\hd(r_C) = 0$, and other vertices initialize them as $\infty$.
On every round every vertex $v$ sends its estimate $\hd(v)$ to all other vertices in the network. Every vertex $u$ sets locally
$$\hd(u) \leftarrow \min \{\hd(u), \min_{v\in V} \{\hd(v) + \omega_{E'^{(k-1)}}(v,u)\}\}~.$$
This computation proceeds for $2\beta+1$ rounds.

As a result every vertex $v$ computes its $(2\beta+1)$-limited distance in $E'^{(k-1)}$ from $\Rt$. For $v$ such that $d_G(\Rt,v) \le 2^{k+1} = \hR^{(k)}$, we have
$$d_G(\Rt, v) ~\le~d^{(2\beta+1)}_{E'^{(k-1)}}(\Rt,v) ~\le~ (1+\zeta)d_G(\Rt,v)~.$$
For every $i \in [0,\ell-1]$,
$$\delta_i \le     \delta_{\ell-1} = \alpha \cdot (1/\eps)^{\ell-1} + 4 R_{\ell-1} \le \alpha \cdot (1/\eps)^{\ell-1} + 8 \alpha (1/\eps)^{\ell-2} = \alpha(1+8\eps)(1/\eps)^{\ell-1} \le \alpha \cdot (1/\eps)^\ell = \hR^{(k)},$$
for $\eps \le 1/10$.
(Recall that on phase $\ell$ there is no superclustering step, and thus the maximum exploration of any superclustering step is $\delta_{\ell-1}$.)

We conclude that for every $v \in \Ball(\Rt,\delta_i)$, its distance estimate $\hd(v)$ satisfies
\begin{equation}
\label{eq:est}
d_G(\Rt,v) \le \hd(v) \le (1+\zeta) \cdot d_G(\Rt,v)~.
\end{equation}
Moreover, this execution of Bellman-Ford  algorithm also constructs a forest $F$ rooted at the centers $\Rt= \{r_C \mid C \in \cS_i\}$ of $i$-level selected clusters, and every vertex $v$ with $\hd(v) < \infty$  knows the root $r_C$ of its tree in $F$.

For any cluster center $r_{C'}$, $C' \in \chP_i \setmns \cS_i$, such that $\hd(r_{C'}) \le \delta'_i = (1+\zeta)\delta_i$,  the algorithm connects $r_C$ with $r_{C'}$ via a hopset edge $e = (r_C,r_{C'})$ (i.e., $H \leftarrow H \cup \{e\}$) of weight $\omega(e) = \hd(r_{C'})$, where $r_C$ is the root of the tree of $F$ to which $r_{C'}$ belongs.
We also create a supercluster rooted at $C$ (more specifically, at $r_C$) which contains all vertices of $C'$ as above.
Observe that (by (\ref{eq:est})), if $d_G(r_C,r_{C'}) \le \delta_i$ then $\hd(r_{C'}) \le (1 + \zeta) \delta_i = \delta'_i$. Then the edge $(r_C,r_{C'})$ will be added to the hopset, and the cluster $C'$ will be superclustered in an $i$-level cluster created around $C$.
(See also inequality (\ref{eq:appr}) below.)

This completes the description of the superclustering step of phase $i$. The interconnection step is conducted in a similar way, but now the Bellman-Ford algorithm is conducted from each cluster center  in
$\URt = \{r_C \mid C \in \chU_i\}$ {\em separately}. (This is in contrast to the superclustering step, on which the Bellman-Ford algorithm is conducted from the set $\Rt$.)
As a result, every vertex $v$ maintains estimates $\hd(r_C,v)$ for cluster centers  $r_C \in \URt$.
On each step $v$ sends all its estimates $\hd(r_C,v)$ which satisfy $\hd(r_C,v) \le \half \delta'_i$ to the entire network. Recall \lemmaref{lm:property}, which implies that whp, no vertex $v$ has more than $O(\deg_i \cdot \log n)$ cluster centers $\{r_C \mid C \in \chU_i\}$ at distance at most $\half \delta'_i$ from it in $E'^{(k-1)}$. Hence, on each step each vertex $v$ needs to send $O(\deg_i \cdot \log n)$ messages to the entire network, and this requires $O(\deg_i \cdot \log n)$ rounds (whp).
Since $\deg_i = O(n^\rho)$ for all $i$, we conclude that the algorithm requires $O(n^\rho \cdot (1/\rho) \log n \cdot \beta)$ time whp. (Because on each phase we perform $2\beta+1$ steps, each of which lasts $O(\deg_i \cdot  \log n)$ rounds, whp.)

As a result, for every pair of clusters $C,C' \in \chU_i$ such that $d_G(r_C,r_{C'}) \le \half \delta_i$, the edge $(r_C,r_{C'})$ is inserted into the hopset.
Indeed, observe that
$$\half \cdot \delta_i \le \half \cdot \delta_\ell \le \alpha \cdot (1/\eps)^\ell (1+8\eps)/2 \le \alpha \cdot (1/\eps)^\ell = \hR^{(k)}~,$$
for $\eps < 1/10$. Hence $d_G(r_C,r_{C'}) \le \hR^{(k)}$.
Denote by $\hd(r_C,r_{C'})$ the  estimate of $d_G(r_C,r_{C'})$ computed by $r_{C'}$. Then we have
\begin{equation}
\label{eq:appr}
\hd(r_C,r_{C'}) ~=~ d^{(2\beta+1)}_{E'^{(k-1)}}(r_C,r_{C'}) ~\le~ (1+\zeta)d_G(r_C,r_{C'}) ~ \le ~\half (1+\zeta) \delta_i ~=~ \half \delta'_i~.
\end{equation}

Next we analyze the properties of the resulting hopset $H = \bigcup_{k > \log \beta -1} H_k$. The size of the hopset is the same as with the centralized algorithm, but in the stretch analysis we incur additional overhead in comparison with the centralized algorithm. The reason is that in the centralized construction every pair of sufficiently close $\chU_i$ cluster centers were interconnected via an edge of length {\em exactly equal }  to the distance in $G$ between them, while now the length of this edge is equal to the distance between them in $E'^{(k-1)}$, i.e., it is a $(1+\zeta)$-approximation of the distance in $G$ between them.

The following lemma is a distributed analogue of Lemma \ref{lm:hop_str}.

\begin{lemma}
\label{lm:hop_str_distr}
For $x,y$ as in Lemma \ref{lm:hop_str} it holds that
\begin{equation}
\label{eq:hop_str_distr}
d^{(h_i)}_{G \cup H_k}(x,y) ~\le~ d_G(x,y)(1+\zeta)(1 + 16c (i-1) \eps) + (1+\zeta) \cdot 8 \cdot \alpha  c \cdot (1/\eps)^{i-1}~,
\end{equation}
with $h_{i+1}  \le (h_i+1)(1/\eps + 2) +2i+5$.
\end{lemma}
\inline Remark: Note that the hopset $H_k$ alone suffices for approximating single-scale distances, i.e., one does not need hopsets $H_j$ with $j < k$ for the inequality (\ref{eq:hop_str_distr}) to hold.
\proof
Recall that the radius of $i$-level cluster in the distributed variant of our algorithm is at most $R'_i \le (1 + \zeta) R_i$, where $R_i \le \alpha \cdot c \cdot (1/\eps)^{i-1}$.
The proof is again by an induction on $i$.
The basis is $i=0$.
\inline Basis: If $d_G(x,y) \le \half \alpha = \half \delta_0$ and all vertices of $\pi(x,y)$ are clustered in $\chU_0$, then there is a hopset edge in $H_k$ between $x$ and $y$, and its weight is at most $(1+ \zeta)d_G(x,y)$, as required.

\inline Step: We follow closely the proof of Lemma \ref{lm:hop_str}.  In inequality ({\ref{eq:short_str}) we get that for a pair $u,v$ of vertices such that all vertices of $\pi(u,v)$ are clustered in $\chU^{(i)}$, it holds that
$$d^{(h_i +1)\cdot \lceil {{d_G(u,v)} \over {\half \alpha \cdot (1/\eps)^i}} \rceil }_{G \cup H_k}(u,v) ~\le~d_G(u,v)(1+\zeta)(1+16c \cdot i \cdot \eps) + 8 \cdot \alpha  c \cdot (1/\eps)^{i-1} \cdot (1+\zeta)~.$$

Now we consider a pair of vertices $x,y$ such that  $d_G(x,y) \le \half \alpha \cdot (1/\eps)^{i+1}$ and such that $\pi(x,y)$ is $\chU^{(i+1)}$-clustered.
Let $w,z_1,r_1,r_2,z_2,w_2,C_1$ and $C_2$ be as in the proof of Lemma \ref{lm:hop_str}.
In an analogous way we conclude that
$$d_{G \cup H_k}^{(h_{i+1})}(x,y) ~\le~ (1+ \zeta)(1 + 16c \cdot i \cdot \eps) d_G(x,y) + (1+\zeta)8 \cdot \alpha  c \cdot (1/\eps)^i ~,$$
with $h_{i+1} \le (h_i+1) (1/\eps + 2) + 2i+5$.
For $\eps < 1/4$, we have $16 \alpha  \cdot c \cdot (1/\eps)^{i-1} \le 4 \cdot \alpha \cdot c \cdot (1/\eps)^i$, and the assertion of the lemma follows.
\QED



For a pair of vertices $u,v \in V$ with $\hR^{(k)}/2 < d_G(u,v) \le \hR^{(k)}$, $\hR^{(k)} = 2^{k+1}$, the additive term of $(1+\zeta)8 \alpha  c \cdot (1/\eps)^{\ell-1}$ in (\ref{eq:hop_str_distr}) can be incorporated into the multiplicative stretch, i.e., we get
$$
d^{(h_\ell)}_{G \cup H_k}(x,y) ~ \le ~ (1+\zeta)\left( (1+ 16c (\ell-1) \eps) d_G(x,y) + 8c \hR^{(k)}\cdot \eps\right)
~ \le ~ (1+\zeta)(1 + 16c \ell \cdot \eps) d_G(x,y)~.
$$
Set now  $\eps' = 16c \cdot \ell \cdot \eps \le 16c (\log (\kappa\rho) +1/\rho) \cdot \eps$.
We get $\beta = O\left( {{\log \kappa + 1/\rho} \over {\eps}}\right)^{\log \kappa + 1/\rho}$, and stretch $(1+\zeta)(1+\eps)$.

Recall that $\zeta = \zeta_{k-1}$ is the value such that $E'_{k-1}$ provides stretch $1+\zeta$.
For the largest $k_0$ such that $\hR^{(k_0)} = 2^{k_0 + 1} \le \beta$, we have $\zeta_{k_0} = 0$. On the next scale we have $\zeta_{k_0+1} = \eps$, and generally, $1 + \zeta_k = (1 + \zeta_{k-1}) (1+ \eps)$, i.e., $1+ \zeta_k = (1 + \eps)^k$.

Hence the overall stretch of our hopset is $(1 + \eps)^{\log \Lambda}$. By rescaling $\eps'' = {\eps \over {2 \log \Lambda}}$, we get stretch $(1+\eps'')^{\log \Lambda} \le 1+ \eps$.
The number of hops becomes
\begin{equation}
\label{eq:beta}
\beta = O\left({{\log \Lambda} \over {\eps}} \cdot (\log \kappa + 1/\rho)\right)^{\log \kappa + 1/\rho}~.
\end{equation}

The expected number of edges in the hopset is $O(n^{1+1/\kappa} \cdot \log \Lambda)$, and its construction time is, whp,  $O(n^\rho/\rho\cdot\log n \cdot \beta\cdot\log \Lambda )$ rounds.

\begin{theorem}
\label{thm:distr_hop}
For any graph $G = (V,E,\omega)$ with $n$ vertices and diameter $\Lambda$, $2\le\kappa\le (\log n)/4$, $1/2>\rho \ge 1/\kappa$, and $0 < \eps < 1$,
our distributed algorithm for the Congested Clique model computes a $(\beta,\eps)$-hopset $H$ with expected size $O(n^{1+1/\kappa} \cdot \log \Lambda)$
in $O(n^\rho/\rho\cdot\log n\cdot\beta\cdot \log \Lambda)$ rounds, whp, with $\beta$ given by (\ref{eq:beta}).

Moreover, a single-scale hopset $H_k$ that provides stretch at most $1+\eps$ using at most $\beta$ hops for pairs $u,v$ with $d_G(u,v) \in (2^k,2^{k+1}]$, and expected size
$O(n^{1+1/\kappa})$, can be computed in the same number of rounds.
\end{theorem}

Next we show that whenever $\Lambda={\rm poly(n)}$, $\beta$ can be made independent of $\Lambda$ and of $n$. Later in \sectionref{sec:reduction} we remove this assumption on $\Lambda$.

Fix a parameter $1\le t\le\log\Lambda$. We partition the set of at most $\log \Lambda$ indices $k$ for which we build hopsets $H_k$ into at most $(\log\Lambda)/t+1$ groups, each consisting of $t$ consecutive indices. Consider a single group $\{k_1,k_1+1,\ldots,k_1+t-1\}$ of indices. (Except maybe one group which may contain less than $t$ indices.)  We will compute all hopsets in this group using the hopset $H^{(k_1)}$, i.e.,  when conducting a Bellman-Ford exploration to depth at most $\delta \le 2^{k+1}$ for some index $k$ in the group, we will conduct this exploration in $G\cup H^{(k_1)}$ for $O(\beta \cdot 2^{k+1 - k_1}) = O(\beta\cdot 2^t)$ rounds.
As a result we spend more time when constructing each individual hopset $H_k$ for $k > k_1$ in the group, but the hopsets that we compute provide a better approximation. (Because they rely on $(1+\zeta_{k_1})$-approximate distances that $G \cup H^{(k_1)}$ provides, rather than on $(1+ \zeta_{k-1})$-approximate distances that $H^{(k-1)}$ provides.)

As a result the ultimate stretch of our hopset becomes just $(1+ \eps)^{(\log\Lambda)/t}$. For a sufficiently small $\eps$, this stretch
is at most $1 +O(\eps\cdot\log\Lambda)/t$.
 We now rescale $\eps' =  (\eps\cdot\log \Lambda)/t$.
Our $\beta$ becomes
\begin{equation}
\label{eq:impr_beta}
\beta =  O\left( {{(\log \kappa + 1/\rho) \log \Lambda}  \over {\eps \cdot t}}\right)^{\log \kappa + 1/\rho}~.
\end{equation}
Finally, the number of rounds becomes greater than it was in Theorem \ref{thm:distr_hop}  by a factor of $2^t$, i.e, it is now $O(n^{\rho}/ \rho\cdot \log n\cdot\beta\cdot \log \Lambda\cdot 2^t)$, and we have the following result.

\begin{theorem}
\label{thm:impr_distr_hop}
For any graph $G = (V,E,\omega)$ with $n$ vertices and diameter $\Lambda$, $2\le\kappa\le (\log n)/4$, $1/2>\rho \ge 1/\kappa$, $1\le t\le\log\Lambda$, and $0 < \eps < 1$, a variant of our distributed algorithm for the Congested Clique model computes a $(\beta,\eps)$-hopset $H$ with expected size $O(n^{1+1/\kappa} \cdot \log \Lambda)$
in $O(n^\rho/\rho\cdot\log n\cdot\beta\cdot \log \Lambda\cdot 2^t)$ rounds whp, with $\beta$ given by (\ref{eq:impr_beta}).

Moreover, a single-scale hopset with expected size $O(n^{1+1/\kappa})$ and with the same $\beta$ can also be constructed within this running time.
\end{theorem}
We note that our algorithm computes a  single-scale hopset of expected size just  $O(n^{1+1/\kappa})$, because hopsets of previous scales are only used to compute (approximate) distances, whereas the stretch analysis only uses the current scale edges and the graph edges.

When $t=1$ this recaptures \theoremref{thm:distr_hop}. A useful assignment is $t=\rho\log n$, which increases the number of rounds by a factor of $n^{\rho}$. Rescaling $\rho'=2\rho$ we obtain the same size and running time as in \theoremref{thm:distr_hop} with
\begin{equation}
\label{eq:beta_constant}
\beta =  O\left( {{(\log \kappa + 1/\rho) \log \Lambda}  \over {\eps \cdot \rho\log n}}\right)^{\log \kappa + 2/\rho}~.
\end{equation}
The constant factor 2 in the exponent can be made arbitrarily close to 1, at the expense of increasing the constant hidden in the O-notation in the base of the exponent.
See also Theorem \ref{thm:clique_path_rep_hopset} for a result about constructing {\em path-reporting} (to be defined) hopsets in the Congested Clique model.

\subsection{Distributed Implementation in CONGEST Model}
\label{sec:congest}

In this section we consider a scenario when we have an underlying "backbone" $n$-vertex network $G = (V,E)$ of hop-diameter $D$, and a "virtual" weighted  $m$-vertex network $\tG = (\tV,\tE,\tomega)$, $\tV \subseteq V$. Our objective is to compute a hopset $H$ for $\tG$.
Observe that the hopset $H$ needs only to approximate distances in $\tG$, defined by the weight function $\tomega$. The latter may have nothing to do with the distance function $d_G$ of $G$.

Constructing hopsets in this framework turns out to be particularly useful for shortest paths computation and routing in CONGEST model
\cite{N14,HKN15,EN16}.

The algorithm itself is essentially the same as in the distributed Congested Clique model, except that most of the communication is conducted via a BFS tree
$\tau$ rooted at a vertex $\rt$ of the backbone network $G$. (The tree has hop-diameter $D$, and it can be constructed in $O(D)$ distributed time.)

Similarly to the Congested Clique model, we construct hopsets $H_{k_0}, H_{k_0+1}, \ldots, H_{\lceil \log \Lambda \rceil}$ one after another, where $k_0$ is the maximum integer $k$ such that $2^{k+1} \le \beta$. ($H_{k_0}$ is set as $\emset$.)
We consider a fixed phase $i$, and describe how the superclustering and interconnection steps of this phase are implemented.
In the superclustering step we conduct Bellman-Ford algorithm from the set $\Rt = \{r_C \mid C \in \cS_i\}$ to depth $2 \beta + 1$ in $\tG \cup H^{(k-1)}$.
Each  vertex $ v \in \tV$ maintains an estimate $\hd(v)$ initialized as 0 if $v \in \Rt$, and $\infty$ otherwise.
For each Bellman-Ford step we collect all the $m$ distance estimates at the root $\rt$ of $\tau$ via pipelined convergecast, and broadcast all these estimates to the entire graph via pipelined broadcast over $\tau$. This requires $O(m + D)$ time. Since we have $O(\beta)$ such steps, overall the superclustering step of phase $i$, for any $i$,  requires $O((m + D) \cdot \beta)$ time.

Next, we implement the interconnection step. Here we need to conduct a Bellman-Ford to hop-depth  at most $2\beta+1$ and to weighted depth at most $\delta'_i/2 = (1+ \zeta)\delta_i/2$ in  $G' \cup H^{(k-1)}$ from all vertices of $\URt$ separately in parallel. As was argued in the previous sections, each vertex $v$  has to maintain $O(\deg_i \cdot \log m) = O(m^\rho \cdot \log m)$ distance estimates, whp.
(In fact, in expectation the total number of these estimates is $O(m \cdot \deg_i) = O(m^{1+\rho})$. Note that only estimates smaller or equal than $\delta'_i/2$ are broadcasted.)

This broadcast is also implemented  via pipelined convergecast and broadcast, and it requires $O(D + m^{1+\rho} \cdot \log m)$  time, whp. Since we do it for $O(\beta)$ steps, we obtain overall  time $O((D + m^{1+\rho} \cdot \log m) \cdot \beta)$  for implementing a single step, and overall time of $O((D+ m^{1+\rho} \cdot \log m) \cdot \beta/\rho)$, for all steps. (Steps of stage 2 require that much time, while the running time of interconnection steps of stage 1 is dominated by this expression.)

Similarly to the Congested Clique model one can decrease the $\beta$ here too by grouping the $\log\Lambda$ scales into groups of size $t$. (Note that we can compute all hopsets within the same group in parallel. This multiplies the number of messages broadcasted in each step by $t$, but the number of steps becomes smaller by the same factor.)
We summarize the result in the following theorem.

\begin{theorem}
\label{thm:congest}
For any graph $G = (V,E,\omega)$ with diameter $\Lambda$ and hop-diameter $D$, any $m$-vertex weighted graph $\tG= (\tV,\tE,\tomega)$ embedded in $G$, any $2\le\kappa\le (\log m)/4$, $1/2 > \rho \ge 1/\kappa$, $1\le t\le\log\Lambda$, $0 < \eps < 1$, our distributed algorithm in the CONGEST model computes a $(\beta,\eps)$-hopset $H$ for $\tG$ with expected size $O(m^{1+1/\kappa} \cdot \log \Lambda)$
in $O((D+ m^{1+\rho} \cdot \log m\cdot t)\beta/\rho\cdot \log \Lambda\cdot 2^t/t)$ rounds whp, with $\beta$ given by (\ref{eq:impr_beta}).
\end{theorem}

\subsection{Path-Reporting Hopsets in Distributed  Models }
\label{sec:aware}

Next we introduce a property of distributed hopset construction which we call {\em awareness}, and argue that our distributed construction satisfies this property.
This property is useful for certain distributed applications of hopsets, such as in constructions of routing tables and sketches, cf. \cite{EN16}.

From this point on this section focuses on the distributed CONGEST model, cf. Section \ref{sec:congest}. At the end of this section we also indicate how these results apply to the Congested Clique model.

The  awareness property stipulates that for every hopset edge $(u,v) \in H$, there exists a path $\pi(u,v)$ in $G$ between $u$ and $v$ of weight $\omega_G(\pi(u,v)) = \omega_H(u,v)$,
 and moreover, every vertex $x$ on this path is {\em aware} that it lies on $\pi(u,v)$ and knows its distances $d_{\pi(u,v)}(u,x)$ and $d_{\pi(u,v)}(x,v)$ to $u$ and to $v$ on this path, respectively,
and its two $G$-neighbors $u'$ and $v'$ that lie on $\pi(u,v)$ along with the orientation. (That is, $x$ knows that $u'$ (resp., $v'$) is its $\pi(u,v)$-neighbor that leads to $u$ (resp., $v$).)
We also call a hopset with this property a {\em path-reporting hopset}, because it can be used to report approximate shortest paths.

Next we adapt our algorithm for constructing hopsets in the CONGEST model so that the awareness property will hold. Assume inductively
that the property holds for hopsets $\bigcup_{j \le k-1} H_j$, and we will now show how to make it hold for $H_k$. (The induction basis holds vacuously, because for the maximum value $k_0$ such that $\hR^{(k_0)} = 2^{k_0+1} \le \beta$, the hopset $H_{k_0} = \emptyset$.)


The first modification to the algorithm that we introduce is that Bellman-Ford executions will not only propagate distance estimates, but also the actual paths that implement these estimates.
First, consider the variant of our algorithm that provides (relatively) large $\beta$, i.e., the $\beta$ given by (\ref{eq:beta}).
Since all our Bellman-Ford invocations in this variant are $(2\beta+1)$-limited, these paths are of length up to $2\beta+1$, and the algorithm incurs a slowdown by only  a factor of $O(\beta)$ as a result of this modification.
Specifically, the Bellman-Ford invocations require now $O((D+ \beta m^{1+\rho}\cdot \log m) \cdot \beta/\rho)$ overall  time (whp), rather than $O((D+ m^{1+\rho}\cdot \log m) \cdot \beta/\rho)$ time, which we had in Section \ref{sec:congest}.

With this modification, when a vertex $r_{C'}$ decides to add an edge $(r_C,r_{C'})$ to the hopset, it knows the entire path $\pi_{k-1}(r_C,r_{C'})$ in $E'^{(k-1)}$ which implements this edge.
(Note that $|\pi_{k-1}(r_C,r_{C'})| \le 2\beta+1$.) After the construction of the hopset $H_k$ is over,  all vertices $v \in \tV$ broadcast all hopset edges (of $H_k$) along with their respective paths to the entire graph. We will refer to this broadcast as the {\em paths' broadcast}.
Since in expectation $|H_k| = O(m^{1+1/\kappa})$, it follows that this broadcast requires  expected $O(m^{1+1/\kappa} \cdot \beta + D)$ time.

For an edge $e = (r_C,r_{C'}) \in H_k$, every vertex $v \in \tV(\pi_{k-1}(r_C,r_{C'}))$ hears this broadcast of $e$ and of $\pi_{k-1}(r_C,r_{C'})$, and writes down to himself that it (i.e., the vertex $v$) belongs to
$\tV(\pi_{k-1}(r_C,r_{C'}))$, calculates its distances to the endpoints $r_C$ and $r_{C'}$, and computes its neighbors $u$ and $u'$ on $\pi_{k-1}(r_C,r_{C'})$ in the direction of $r_C$ and $r_{C'}$, respectively.

Now vertices $x$ involved in a path $\hpi(v,u) \subseteq \tE$ that implements an edge $(v,u)$ of $\pi_{k-1}(r_C,r_{C'})$ need also to write down that they belong to the path $\hpi(r_C,r_{C'})$ which implements the $H_k$-edge $e = (r_C,r_{C'})$ via edges of $\tE$.
%
%
(This is in contrast to $\pi_{k-1}(r_C,r_{C'})$, which implements the same $H_k$-edge $(r_C,r_{C'})$ via edges of $E'^{(k-1)}$.)
Since $x$ hears of the edge $(v,u)\in \hpi(r_C,r_{C'})$, and since $v$ stores the distances to $v$ and $u$, it can infer from $\hpi(r_C,r_{C'})$ the distances to $r_C,r_{C'}$ (e.g., if $v$ is the endpoint closer to $r_C$, then the distance that $x$ stores to $r_C$ is $d_{\pi(r_C,r_{C'})}(r_C,v)+d_{\pi(v,u)}(v,x)$). The appropriate neighbors of $x$ are the same as those it stores for the edge $(v,u)$.


To summarize, this modification of the algorithm ensures that our hopset constructing algorithm satisfies the awareness property. It does so by having its running time increased to
$O((m^{1+1/\rho} \cdot \log m \cdot  \beta + D) \beta/\rho \log\Lambda)$ rounds, whp. (The time to broadcast $O(m^{1+1/\kappa} \beta)$ messages is dominated by $m^{1+\rho} \cdot \log m \cdot \beta$ here.) For the variant in which we group $t$ scales together, each path may consist of $O(\beta\cdot 2^t)$ hops, so we pay this  factor in the number of messages sent.

\begin{theorem}
\label{thm:hop_aware_small_beta}
For any graph $G = (V,E,\omega)$ with diameter $\Lambda$ and hop-diameter $D$, $2\le\kappa\le (\log n)/4$, $1/2 > \rho \ge 1/\kappa$, $1\le t\le\log\Lambda$, $0 < \eps \le 1$, and any $m$-vertex weighted graph $\tG= (\tV,\tE,\tomega)$ embedded in $G$, our distributed algorithm in the CONGEST model computes a {\em path-reporting} $(\beta,\eps)$-hopset $H$ for $\tG$ with expected size $O(m^{1+1/\kappa} \cdot \log \Lambda)$
in $O((D+ m^{1+\rho} \cdot \log m\cdot t\cdot\beta\cdot 2^t)\beta/\rho\cdot \log \Lambda\cdot 2^t/t)$ rounds whp, with $\beta$ given by (\ref{eq:impr_beta}).

Moreover, it can also compute a single-scale hopset with expected $O(m^{1+1/\kappa})$ edges, and the same $\beta$ and within the same running time.
\end{theorem}

A similar adaptation enables us to construct path-reporting hopsets in the Congested Clique model, again, by transmitting entire paths rather than distance estimates. As paths are of length at most $\beta\cdot 2^t$, this incurs such a factor in the number of rounds. Specifically, we obtain the following result.

\begin{theorem}
\label{thm:clique_path_rep_hopset}
For any graph $G = (V,E,\omega)$ with $n$ vertices and diameter $\Lambda$, $2\le\kappa\le (\log n)/4$, $1/2 > \rho \ge 1/\kappa$, $1\le t\le\log\Lambda$, and $0 < \eps \le 1$, a variant of our distributed algorithm for the Congested Clique model computes a {\em path-reporting} $(\beta,\eps)$-hopset $H$ with expected size $O(n^{1+1/\kappa} \cdot \log \Lambda)$
in $O(n^\rho/\rho\cdot\log n\cdot\beta^2\cdot \log \Lambda\cdot 2^{2t})$ rounds, with $\beta$ given by (\ref{eq:impr_beta}).

Moreover, a single-scale hopset with expected size $O(n^{1+1/\kappa})$ and with the same $\beta$ can also be constructed within this running time.
\end{theorem}

A particularly useful setting of the parameter $t$ is $t = \log (m^\rho)$ in the CONGEST model, and $t = \log (n^{\rho/2})$ in the Congested Clique model.
Then in Theorem \ref{thm:hop_aware_small_beta} we obtain running time $O((D + m^{1+2\rho} \cdot \log^2 m \cdot \beta \rho) \beta/\rho^2 \cdot m^\rho)$,
and
\begin{equation}
\label{eq:beta_impr}
\beta ~=~ O\left({{\log \kappa+1/\rho} \over {\eps \cdot \rho}}\right)^{\log \kappa + 1/\rho}~,
\end{equation}
 assuming $\Lambda \le \poly(m)$.

In Theorem \ref{thm:clique_path_rep_hopset} we obtain running time $O(n^{2\rho}/\rho \cdot \log^2 n \cdot \beta^2)$ with the same value of $\beta$ as in (\ref{eq:beta_impr}), assuming $\Lambda \le \poly(n)$. One can also rescale $\rho' = 2\rho$, and get running time $O(n^\rho/\rho \cdot \log^2 n \cdot \beta^2)$,
with
$$\beta = O\left({{\log \kappa + 1/\rho} \over {\eps \cdot \rho}}\right)^{\log \kappa + 2/\rho}~.$$
We remark that the constant factor 2 in the exponent can be made arbitrarily close to 1, at the expense of increasing a constant factor hidden in the $O$-notation in the base of the exponent. Also the assumption that $\Lambda \le \poly(n)$ will be removed in \sectionref{sec:reduction}.

\subsection{Streaming Model}
\label{sec:stream}

Our implementation of the hopset-constructing algorithm in the streaming model follows closely our implementation from Section \ref{sec:hopset_clique} of the algorithm in the Congested Clique model.

Here too we construct the hopsets $H_{k_0},H_{k_0+1},\ldots,H_{ \lceil \log \Lambda \rceil}$, one after another.
The hopset $H_{k_0} = \emptyset$.
Next we describe how to construct a hopset $H_k$ (for distances in the range $(2^k,2^{k+1}]$, $2^{k+1} = \hR^{(k)}$), assuming that the hopsets $H_j$, $j \le k-1$, were already constructed.

We describe the streaming implementation in two regimes. In the first regime the space will be $\tO(n^{1+\rho})$, but the number of passes will be $\polylog(n) \cdot \beta$.
 (In fact, in this section it will be $\polylog(\Lambda) \cdot \beta$, but in \sectionref{sec:reduction} we will replace the dependence on $\Lambda$ by a similar dependence on $n$.) In the second regime we use
$O(n^{1+1/\kappa} \cdot \log \Lambda)$ space, but the number of passes is much larger. (Specifically, it is $O(n^\rho \cdot \beta \cdot \log n \cdot \log \Lambda)$.)

To conduct the  superclustering step of a phase $i$ we use every pass over the stream $E$ of edges of $G$ to update distance estimates $\hd(v)$ of the distances
$d_G(\Rt,v)$. After each pass we also read again the hopset $H^{(k-1)}$, and adjust distance estimates according to hopset edges as well.
As a result of $2 \beta + 1$ such passes  we implement a $(2\beta+1)$-limited Bellman-Ford algorithm in $E'^{(k-1)} = E \cup H^{(k-1)}$, originated at $\Rt$.
As was argued above, this provides us with a $(1 + \zeta_{k-1})$-approximate distances $d_G(\Rt,v)$, for all $v$ such that $d_G(\Rt,v) \le \delta_i$.
This completes the description of the superclustering step. The space required for it is the space needed to keep the hopset, i.e., expected $O(n^{1+1/\kappa} \cdot \log \Lambda)$, and $O(n)$ space for distance estimates.

For the interconnection step we conduct a similar $(2\beta+1)$-limited Bellman-Ford exploration, but to depth at most $\delta'_i/2 = {{\delta_i(1 + \zeta_{k-1})} \over 2}$, and from all cluster centers $\URt$ separately in parallel. This necessitates each vertex to maintain in expectation $O(\deg_i ) = O(n^\rho)$ estimates, so we will use $O(n^{1+\rho}/\rho)$ space to guarantee (with constant probability) that none of the $1/\rho$ phases of stage 2 overflows (note that in stage 1 $\deg_i$ is much smaller than $n^{\rho}$). This completes the description of the streaming algorithm in the first regime. The space required is $O(n^{1+\rho}/\rho+ n^{1+1/\kappa} \cdot \log \Lambda)$, and the number of passes is  $O(\beta \cdot \log \Lambda)$.

For applications in which we use hopsets to provide approximate paths rather than distances, we can store actual paths in $E'^{(k-1)}$ for every edge of $H_k$ that we create.
A hopset appended with this information will be referred to as a {\em path-reporting hopset}.
Since every edge of $H_k$ is implemented using at most $2 \beta+1$ edges of $E'^{(k-1)}$, we can construct the path-reporting variant of the above hopset
using space $O(n^{1+\rho}/\rho + n^{1+1/\kappa} \cdot \beta \cdot \log \Lambda)$, and the same number of passes as above.

Next, consider the regime when we allow space  $O(n^{1+1/\kappa} \cdot \log \Lambda)$.
The superclustering steps can be implemented  in the same way as was described above. The interconnection steps however require certain adaptation.
Specifically, partition the interconnection step of phase $i$ to $c \cdot \deg_i \cdot \log n$ subphases, for a sufficiently large constant $c$.
On each subphase each exploration source, which was not sampled on previous subphases, samples itself i.a.r. with probability $1/\deg_i$. Then the sampled exploration sources conduct the $\delta'_i/2$-distance-bounded $(2\beta+1)$-limited Bellman-Ford explorations.
Recall \lemmaref{lm:property}, that asserts whp every vertex is visited by at most $O(\deg_i\cdot\log n)$ explorations in each phase. Since in each subphase every exploration happens with probability $1/\deg_i$, Chernoff bound implies that whp no vertex is visited by more than $O(\log n)$ explorations. We conclude that it suffices to use $O(n\cdot \log n)$ memory for all phases to keep distance estimates.  Here we take union-bound on all the $\log \Lambda$ different scales too, assuming that $\log \Lambda  \le \poly(n)$.
After $c \cdot \deg_i \cdot  \log n$ subphases, whp, each exploration source is sampled on at least one of the subphases, and so the algorithm performs all the required explorations.

So we have space $O(n^{1+1/\kappa} \cdot \log \Lambda)$ and number of passes is $O(n^\rho \cdot \log n \cdot \beta\cdot\log\Lambda)$.
The size, stretch and hopbound analysis of the resulting hopset is identical to the one we had in the distributed Congested Clique model.


We can reduce the value of $\beta$ by employing the idea used for the proof of \theoremref{thm:impr_distr_hop}.
(See the discussion right before \theoremref{thm:impr_distr_hop}.)
Specifically, fix some integer $t\ge 1$. The set of at most $\log \Lambda$ indices $k$ for which we build hopsets $H_k$ is partitioned into at most $(\log \Lambda)/t+1$ groups, each consisting of $t$ consecutive indices  (except maybe one of them which may consist of less than $t$ indices). In a single group $\{k_1,k_1+1,\ldots,k_1+t - 1\}$ of indices, each hopset $H_k$ from the group is computed using hopset $H^{(k_1)}$, rather than $H^{(k-1)}$.

As a result, we now conduct Bellman-Ford explorations to depth $O(\beta \cdot 2^t)$, rather than just $O(\beta)$, hence the number of passes increases by a factor of $2^t$. Hopsets in the same group can be computed "in parallel", so the space increases by a factor of $t$. (Note that in the path-reporting setting, we need to store paths of length $\beta\cdot 2^t$, which increases the space needed to store the hopset by this factor.) The hopbound $\beta$ improves to the value in \eqref{eq:impr_beta}.
We summarize this discussion in the next theorem.

\begin{theorem}
\label{thm:stream_hop-alt}
For any $n$-vertex graph $G = (V,E,\omega)$ of diameter $\Lambda$, any $2\le\kappa\le(\log n)/4$, $1/2 > \rho \ge 1/\kappa$,  $0 < \eps \le 1$, and any $1\le t\le\log\Lambda$,
our streaming algorithm computes an $(\beta,\eps)$-hopset with $\beta$ given by (\ref{eq:impr_beta}) and with expected size $O(n^{1+1/\kappa}\log \Lambda)$.
The resource usage is either
\begin{enumerate}
\item
space $O(t\cdot n^{1+\rho} /\rho + n^{1+1/\kappa} \cdot \log \Lambda))$
(resp., space $O(t\cdot n^{1+\rho}/\rho + n^{1+1/\kappa} \cdot \beta\cdot 2^t\cdot \log \Lambda)$ for path-reporting), in expectation,
and $O(\beta \log \Lambda\cdot 2^t)$ passes, or
\item
\label{item:small_space}
space $O(n^{1+1/\kappa} \cdot \log \Lambda)$
(resp.,  $O(n^{1+1/\kappa} \cdot  \beta \cdot 2^t\cdot \log \Lambda)$ for path-reporting) in expectation,
and $O(n^\rho \cdot \beta \cdot \log n \cdot \log \Lambda\cdot 2^t)$ passes.\footnote{Here we assume $\Lambda \le 2^{\poly(n)}$.}
\end{enumerate}
Moreover, a single-scale hopset $H_k$ with expected size $O(n^{1+1/\kappa})$ can be computed within the space and pass complexities stated above.
\end{theorem}

In item  \ref{item:small_space} of Theorem \ref{thm:stream_hop-alt} in the path-reporting case, it makes sense to set $t = \log(n^{1/\kappa})$. As a result we obtain space
$O(n^{1+ 2/\kappa} \cdot \beta \cdot \log \Lambda)$, in expectation, and $O(n^{\rho + 1/\kappa} \cdot \log \Lambda)$ passes, and $\beta$ is given by (\ref{eq:beta_impr}). One can also rescale $\kappa' = \kappa/2$, and get expected space $O(n^{1+1/\kappa} \cdot \beta \cdot \log \Lambda)$, $O(n^{\rho + {1 \over {2\kappa}}})$ passes, and $\beta  = O\left({{(\log \kappa+1/\rho) \kappa} \over {\eps}}\right)^{\log\kappa+ 1/\rho +1}$.

\subsection{PRAM Model}
\label{sec:pram_hopset}

We construct hopsets $H_{k_0},H_{k_0+1},\ldots,H_\lambda$, $\lambda = \lceil \log \Lambda \rceil$, one after another. Suppose that the hopset $H^{(k-1)} = \bigcup_{j=k_0}^{k-1} H_j$ has already been constructed. We now construct the hopset $H_k$ for distances in the range $(2^k,2^{k+1}]$, $2^{k+1} = \hR^{(k)}$.

We designate a set $P_v = \{p_{v,1},\ldots,p_{v,\Delta}\}$ of $\Delta = c \cdot n^\rho \cdot \log n$ processors for every vertex $v \in V$, and a set $P_{e} = \{p_{e,1},\ldots,p_{e,\Delta}\}$ of $\Delta$ processors for every edge $e \in E'^{(k-1)} = E \cup H^{(k-1)}$.

Next, we describe how to implement the superclustering step of a phase $i$, and later we will explain how the interconnection step is implemented.

For  the superclustering step we use just one processor $p_v \in P_v$ for each vertex $v \in V$, and one processor $p_e \in P_e$ for each edge $e \in E'^{(k-1)}$.
In the superclustering step we run $(2\beta+1)$-limited Bellman-Ford algorithm in $E'^{(k-1)}$, originated at $\Rt$. At the beginning of an iteration of the Bellman-Ford algorithm, for every vertex $v \in V$, its processor $p_v$ maintains an estimate $\td(v)$ of its distance from $\Rt$, and if $\td(v) < \infty$, then $p_v$ also stores the identity of a root $r_C \in \Rt$, such that $\td(v)$ reflects the length of a path from $r_C$ to $v$.\footnote{When we say that a processor $p$ stores a value of a variable  $x$, we mean that there is a memory location $x$, designated to the processor $p$, from which $p$ can read the value of $x$. Other processors can also read from and write to this location.}
 In the path-reporting case, $p_v$ also stores an edge $(u,v)$ through which $v$ acquired this estimate.

To implement the iteration, for each edge $e = (u,v) \in E'^{(k-1)}$ incident on $v$, the processor $p_e$ computes $\td(u) + \omega(e)$. (For this end, all processors $\{p_e \mid u \in e\}$ need to read $\td(u)$ concurrently. This however can be implemented in EREW model in $O(\log n)$ time, cf. \cite{Jaja92}, Theorem 10.1.).
The minimum $\min \{\td(u) + \omega(u,v) \mid u \in \Gamma(v)\}$ can now be computed by the processors $\{p_e \mid v \in e\}$ in $O(\log n)$ parallel time.
If this minimum is smaller than the current value of $\td(v)$, then the estimate $\td(v)$ is updated to be equal to this minimum. Hence the total EREW parallel time for one iteration of Bellman-Ford is $O(\log n)$, and the overall time for the superclustering step is, therefore, $O(\beta \cdot \log n)$.

Now, we turn to implementing the interconnection step. Here we implement a $(2 \beta+1)$-limited Bellman-Ford exploration, to depth at most $\delta'_i/2$, from all cluster centers $\URt$ separately, in parallel.

Consider a single iteration of the Bellman-Ford algorithm. At the beginning of the iteration, every vertex $v$ maintains estimates $\{ \td(v,x) \mid x \in \URt\}$, for all $x \in \URt$ that it heard from. Other estimates are (implicitly) set to $\infty$. For every edge $e = (u,v) \in E'^{(k-1)}$, incident on $v$, whp, there are at most $\Delta = c \cdot n^\rho \cdot \log n$ Bellman-Ford explorations that traverse this edge. Recall that we we have $\Delta$ processors $\{ p_{e,1},\ldots,p_{e,\Delta}\} = P_e$ designated to this edge. We designate a separate processor from $P_e$ to each exploration that traverses $e$. With some notational ambiguity, we will denote by $p_{e,x}$ the processor from $P_e$ designated to the exploration originated at a vertex $x$, traversing the edge $e$.

All processors $\{p_{e,x} \mid e = (v,u), \mbox{~for~some~} u\}$ read the value $\td(u,x)$. (This concurrent read can be implemented in $O(\log n)$ time in EREW PRAM.) They also calculate the value $\td(u,x) + \omega(u,v)$, and the minimum of these values (separately for each $x$) is computed within additional $O(\log n)$ EREW PRAM time. If this minimum is smaller than the current $\td(v,x)$, and if it is no greater than $\delta'_i/2$, then the value $\td(v,x)$ is updated (by the processor $p_{v,x}$, designated to handle at $v$ the exploration originated at $x$) to the new value.

To summarize, one itertaion of Bellman-Ford explorations in an interconnection step can also be implemented in $O(\log n)$ EREW PRAM time. Hence, overall, the superclustering and the interconnection steps of a given phase require $O(\beta \cdot \log n)$ EREW PRAM time. Therefore, the total parallel time for computing a single-scale hopset $H_k$ is $O(\beta \cdot \ell \cdot \log n) = O(\beta \cdot (\log \kappa+ 1/\rho) \cdot \log n)$. Computing hopsets for all the $\lceil \log \Lambda \rceil$ scales requires $ O(\beta \cdot (\log \kappa+ 1/\rho) \cdot \log n \cdot \log \Lambda)$ parallel time.
The number of processors is $O(|E'^{(\lambda)}| \cdot n^\rho \cdot \log n) = O((|E| + |H^{(\lambda)}|) \cdot n^\rho \cdot \log n)$, and
$\Expect(|H^{(\lambda)}|)  = O(n^{1 + 1/\kappa} \cdot \lambda)$. Each single-scale hopset has  size $O(n^{1+1/\kappa} \cdot \log n)$, whp, i.e., $H^{(\lambda)}$ has size $O(n^{1+1/\kappa} \cdot \log n \cdot \log \Lambda)$, whp.

\begin{theorem}
\label{thm:pram_hopset}
For any $n$-vertex graph $G = (V,E,\omega)$ of diameter $\Lambda$, any $2\le\kappa\le(\log n)/4$, $1/2 > \rho \ge 1/\kappa$,  $0 < \eps \le 1$,
our parallel  algorithm computes an $(\beta,\eps)$-hopset with $\beta$ given by (\ref{eq:beta}), and with expected size $O(n^{1+1/\kappa}\cdot \log \Lambda)$,
in $O(\beta \cdot (\log \kappa+ 1/\rho) \cdot \log n \cdot \log \Lambda)$ EREW PRAM time, using $O((|E| + n^{1+1/\kappa} \cdot \log n \cdot \log \Lambda) \cdot n^\rho \log n)$ processors, whp.
\end{theorem}

More generally, we can also group scales into $\lceil {\lambda \over t} \rceil$ groups, of size $t$ each (except maybe one of them, which might be smaller), for a parameter $t$. We then compute all hopsets in a group via explorations in the lowest-scale hopset of that group. As a result, the explorations become $O(2^t \cdot \beta)$-limited, instead $(2\beta+1)$-limited,  but $\beta$ decreases. (It is now given by (\ref{eq:impr_beta}).)  The number of processors grows by a factor of $t$, because all hopsets in the same group are now computed in parallel.

\begin{theorem}
\label{thm:pram_gen_hopset}
For any $n$-vertex graph $G = (V,E,\omega)$ of diameter $\Lambda$, any $2\le\kappa\le(\log n)/4$, $1/2 > \rho \ge 1/\kappa$,  $0 < \eps \le 1$, and any $1 \le t \le \log \Lambda$,
our parallel  algorithm computes an $(\beta,\eps)$-hopset with $\beta$ given by (\ref{eq:impr_beta}) and with expected size $O(n^{1+1/\kappa}\cdot \log \Lambda)$,
in $O(\beta \cdot (\log \kappa+ 1/\rho) \cdot \log n \cdot 2^t \cdot {{\log \Lambda} \over t})$ EREW PRAM time, using $O((|E| + n^{1+1/\kappa} \cdot \log n \cdot \log \Lambda) \cdot n^\rho \log n \cdot t)$ processors, whp.

Moreover, a single-scale hopset of expected size $O(n^{1+1/\kappa})$ can be computed using the same resources.
\end{theorem}

Naturally, Theorem \ref{thm:pram_gen_hopset} can be made path-reporting, by keeping for every hopset edge a path (of length $O(2^t \cdot \beta)$) of lower-scale hopset edges that implement it.

\section{Eliminating Dependence on the Aspect Ratio}
\label{sec:reduction}

In this section, we show a general reduction that removes the dependence on the aspect ratio of the graph, from both the running time and the hopset size.

Assume, without loss of generality, that the minimal distance in the graph $G=(V,E)$ is 1. Fix a parameter $0 < \eps < 1/2$. For any scale index
$k\ge 1$ we define a graph $G_k$, that contains the edges of weight at most $2^{k+2}$, and that every edge of weight less than $(\eps/n)\cdot 2^k$ is contracted. By contraction we mean identifying the edge endpoints while keeping the shortest edge among parallel edges. We refer to the vertices of $G_k$ as {\em nodes}, where each node is a subset of $V$. The weight of an edge $(X,Y)\in E(G_k)$ is set to be
\begin{equation}\label{eq:edge-set}
\cW(X,Y)=\omega(x,y)+(\eps/n)\cdot 2^k\cdot(|X|+|Y|)~,
\end{equation}
where $x\in X$, $y\in Y$, and the edge $(x,y)\in E$ is the shortest edge between a vertex of $X$ to a vertex of $Y$. (The purpose of the additional term $(\eps/n)\cdot 2^k\cdot(|X|+|Y|)$ is to guarantee that distances in $G_k$ are no shorter than those in $G$, while ensuring that for pairs of distance $\ge 2^k$, the distance in $G_k$ does not increase by too much).

In order to guarantee a small number of hops even for contracted vertices, we shall add an additional set of edges $S$ to the hopset. Every node $U$ in $G_k$ has a designated center $u\in U$, and we add edges from $u$ to every vertex in $U$ to the hopset. 
Consider a contraction of an edge $(x',y')$, $x' \in X$, $y' \in Y$, connecting nodes $X,Y$, with centers $x\in X$ and $y\in Y$.  Assuming $|X|\ge |Y|$, then $x$ is declared the center of $U=X\cup Y$, and we add to $S$ edges from $x$ to every vertex of $Y$. The weight of the edge $(x,z)$ for each $z\in Y$ is set as
\begin{equation}\label{eq:setw}
\cW(x,z)=(\eps/n)\cdot 2^k\cdot |U|~.
\end{equation}
This value dominates $d_G(x,z)$, as there exists a path from $x$ to $z$ consisting of at most $|U|-1$ edges, each of weight at most $(\eps/n)\cdot 2^k$.
\begin{claim}\label{claim:sizeS}
$S\le n\log n$.
\end{claim}
\proof Assume inductively that every node $U$ of size $s=|U|$ has at most $s\log s$ internal edges added to the hopset by the process. This holds for singletons $|U|=1$, which have 0 internal edges. When we combine $X$ and $Y$, of sizes $s_1,s_2$, we add at most $s_2=\min\{s_1,s_2\}$ edges. By induction there were already at most $s_1\log s_1+s_2\log s_2$ edges, so the total number of edges in $S$ between vertices of $U$ is at most
\[
s_1\log s_1+s_2\log s_2+s_2= s_1\log s_1+s_2\log (2s_2)\le s_1\log (s_1+s_2)+s_2\log (s_1+s_2)=(s_1+s_2)\log(s_1+s_2)~.
\]
When the scale index $k$ is sufficiently large we have at a certain point a graph with a single node $V$, and at this point we added at most $n\log n$ edges throughout the process.
\QED

\begin{claim}\label{claim:GK}
Let $x,y\in V$ such that $d_G(x,y)\in(2^k,2^{k+1}]$, let $X,Y\in V(G_k)$ be the two nodes containing $x,y$ (respectively) in $G_k$, then
\[
d_G(x,y)\le d_{G_k}(X,Y)\le (1+2\eps)d_G(x,y)~.
\]
\end{claim}
\proof
We start with the right-hand-side inequality.
Let $x=x_0,\dots,x_q=y$ be the shortest path $P$ in $G$ between $x,y$.
Let $X=X_0,X_1,\dots,X_q=Y$ be the corresponding nodes in $G_k$ (that is, $x_j\in X_j$), and let $X=X_0,\dots,X_p=Y$ be that path with all repetitions and loops removed. For $0\le j\le p$, denote by $s(j)$ (resp. $t(j)$) the index of the first (resp., last) vertex of $P$ that $X_j$ contains.
Since, by (\ref{eq:edge-set}),  the weight of each edge $(X_{j-1},X_j)$, for $1 \le j \le p$,  in $G_k$, is defined using the shortest weight edge, we have  $\cW(X_{j-1},X_j)\le \omega(x_{t(j-1)},x_{s(j)})+(\eps/n)\cdot 2^k\cdot(|X_{j-1}|+|X_j|)$. As each term $|X_j|$ appears at most twice, and $\sum_{j=0}^p|X_j|\le n$, we obtain that
\begin{eqnarray*}
d_{G_k}(X,Y)&\le& \sum_{j=1}^p \cW(X_{j-1},X_j)\\
&\le& \sum_{j=1}^p (\omega(x_{t(j-1)},x_{s(j)})+(\eps/n)\cdot 2^k\cdot(|X_{j-1}|+|X_j|)) \\
&\le& d_G(x,y)+2\eps\cdot 2^k\\
&\le&(1+2\eps) \cdot d_G(x,y)~.
\end{eqnarray*}
We now turn to prove the left-hand-side inequality. Let $X=Y_0,\dots,Y_r=Y$ be the shortest path in $G_k$ from $X$ to $Y$.
For each $1 \le j \le r$,
denote by $(y_{j-1},z_j)\in E$ the edge of minimal weight connecting $Y_{j-1}$ and $Y_j$, with $y_{j-1}\in Y_{j-1}$ and $z_j\in Y_j$. Since each $Y_j$ consists of $|Y_j|-1$ edges that were contracted, each of weight at most $(\eps/n)\cdot 2^k$, we have that $d_G(z_j,y_j)\le{\rm diam}(Y_j)\le (\eps/n)\cdot 2^k\cdot|Y_j|$.
Moreover, this inequality holds for every pair $z'_j,y'_j$ of vertices in $Y_j$.
Hence
\begin{eqnarray*}
d_{G_k}(X,Y)&=&\sum_{j=1}^r\cW(Y_{j-1},Y_j)\\
&\stackrel{(\ref{eq:edge-set})}{=}&\sum_{j=1}^r\left[\omega(y_{j-1},z_j)+(\eps/n)\cdot 2^k\cdot(|Y_{j-1}|+|Y_j|)\right]\\
&\ge&\sum_{j=1}^r\left[d_G(y_{j-1},z_j)+d_G(z_j,y_j)\right]+d_G(x,y_0)+d_G(z_r,y)\\
&\ge& d_G(x,y)~.
\end{eqnarray*}
\QED

Some of the scales $k$ are redundant -- define $K$ to be the set of scales $k$ so that there exists an edge of weight in the range $[2^k/n,2^{k+1}]$.
We will refer to the scales in $K$ as {\em relevant} scales.
Observe that if there is no edge in this range, then there is no pair of vertices whose distance in $G$ is in the range $(2^k,2^{k+1}]$, so we do not need a hopset for this scale. We can see that $|K|\le \tilde{O}(|E|)$, as every edge can induce a logarithmic number of scales to $K$.

For every $k\in K$ and every connected component of $G_k$, we will execute the algorithm for constructing a single scale $(\beta,\eps)$-hopset $H_k$ as in \theoremref{thm:hop_exist}. Whenever we add a hopset  edge between two nodes $X,Y$, we put the same hopset edge between their centers.
\begin{lemma}
The set $H=S\cup\bigcup_{k\in K}H_k$ is a $(6\beta+5,6\eps)$-hopset for $G$.
\end{lemma}
\proof
Fix $x,y\in V$, and let $k\in K$ be the scale such that $d_G(x,y)\in(2^k,2^{k+1}]$. Let $X,Y\in V(G_k)$ so that $x\in X$ and $y\in Y$. Since $H_k$ is a $(\beta,\eps)$-hopset for the range $(2^k,2^{k+1}]$ in $G_k$, it is also a $(2\beta+1,\eps)$-hopset for the range $(2^k,2^{k+2}]$, and by \claimref{claim:GK} we have that indeed $2^k<d_{G_k}(X,Y)\le (1+2\eps)2^{k+1}<2^{k+2}$. It follows that
there exists a path $(X=X_0,X_1,\dots,X_p=Y)$ in $G_k\cup H_k$ containing $p\le 2\beta+1$ edges, of length at most $(1+\eps)d_{G_k}(X,Y)$. For every $X_j$, denote by $u_j$ its center.
Note that each edge $(X_{j-1},X_j)$ could be either a hopset edge or an edge of $G_k$. In the latter case there are vertices $y_{j-1}\in X_{j-1}$ and $x_j\in X_j$ so that $(y_{j-1},x_j)\in E$, and the weight of the edge by
\eqref{eq:edge-set} is $\hat{\omega}(X_{j-1},X_j)=\omega(y_{j-1},x_j)+(\eps/n)\cdot 2^k\cdot(|X_{j-1}|+|X_j|)$.
For ease of notation, in the former case (a hopset edge) we write $y_{j-1}=u_{j-1}$ and $x_j=u_j$, and let $\hat{\omega}(y_{j-1},x_j)$ denote the weight of the hopset edge (recall that this edge indeed connects nodes' centers $u_{j-1},u_j$). Then the following is a path from $x$ to $y$ in $G\cup H$:
\[
P=(x=x_0,u_0,y_0,x_1,u_1,y_1,x_2,u_2,y_2,\dots,x_p,u_p,y_p=y)~.
\]
First note that the path contains at most $2(p+1)$ edges of $S$ (that are inside the nodes), and $p$ edges between nodes. Since $p\le 2\beta+1$, this path has at most $6\beta+5$ edges. Next we bound the stretch. We have that the total weight of the edges in $P\cap S$ is
\[
\sum_{j=0}^p(\hat{\omega}(x_j,u_j)+\hat{\omega}(y_j,u_j))\stackrel{(\ref{eq:setw})}{=}2\sum_{j=0}^p (\eps/n)\cdot 2^k\cdot |X_j|\le (\eps/n)\cdot 2^{k+1}\cdot n\le 2\eps\cdot d_G(x,y)~.
\]
We noted that the length of the path $X_0,X_1,\dots,X_p$ is
\[
\sum_{j=1}^p \hat{\omega}(y_{j-1},x_j)\le(1+\eps)d_{G_k}(X,Y)\le(1+\eps)(1+2\eps)d_G(x,y)~,
\]
where the last inequality is by \claimref{claim:GK}. Combining these inequalities implies that the length of $P$ is at most
\[
2\eps\cdot d_G(x,y)+(1+4\eps)d_G(x,y)=(1+6\eps)d_G(x,y)~.
\]
\QED

We say that a node $U$ in the graph $G_k$ is {\em active} if it has degree at least 1, and denote by $n_k$ the number of active nodes in $G_k$
\begin{claim}\label{claim:active-nodes}
$\sum_{k\in K}n_k=O(n\log n)$.
\end{claim}
\proof
The nodes of the graphs $\{G_k\}_{k\in K}$ induce a laminar family $\cL$ on $V$, which contains at most $2n-1$ distinct sets. In order to bound the total number of active nodes in all these graphs, it suffices to show that each node can be active in at most $\log(n/\eps)+2$ scales. To this end, consider a node $U$ that is active for the first time in $G_k$, so it has an edge containing it of weight at most $2^{k+2}$.
(That is, $k$ is the smallest scale such that $U$ is active in $G_k$.)
After $q=\log(n/\eps)+2$ scales, in $G_{k+q}$, this edge will be of weight at most $(\eps/n)\cdot 2^{k+q}$. Thus it will be contracted, and the node $U$ will merge with some other node and never appear again in $G_{k'}$ for $k'\ge k+q$.
\QED

We will refer to $\cL$ as the {\em laminar family} of the algorithm.

\begin{corollary}
$|H|=O( n^{1+1/\kappa}\cdot\log n)$, in expectation.
\end{corollary}
\proof
By \claimref{claim:sizeS}, we have that $|S|\le n\log n$. For each $k\in K$, by \theoremref{thm:hop_exist}, the expected size of the single-scale hopset $H_k$ is at most $O(n_k^{1+1/\kappa})$ (observe that isolated nodes do not participate in the hopset). We conclude by \claimref{claim:active-nodes} that
\[
|H|\le |S|+\sum_{k\in K}|H_k|\le n\log n + \sum_{k\in K}O(n_k^{1+1/\kappa})\le n\log n+n^{1/\kappa}\sum_{k\in K}O(n_k)=O(n^{1+1/\kappa}\cdot\log n)~.
\]
\QED

\subsection{Implementation in the Centralized Model}

Computing the graphs $G_k$ can be done in $\tilde{O}(|E|)$ time in a straightforward manner. First sort the edges by weight, add all edges of weight at most 2 to obtain $G_0$, and create $G_k$ from $G_{k-1}$ by adding edges of weight in the range $(2^k,2^{k+1}]$, and contracting those of weight in range $(\eps/n)\cdot(2^{k-1},2^k]$. While computing the graphs we store all the nodes, so we may add to $S$ all the relevant edges.

Note that for any $k\in K$, the aspect ratio of $G_k$ is $O(n/\eps)$. By \theoremref{thm:hop_exist}, the  expected running time for computing the hopset $H_k$ is $O(|E(G_k)|+n_k\log n)\cdot n^\rho/\rho$. Observe that each edge participates in at most $\log(n/\eps)+2$ scales.
We have
  $\sum_{k\in K}|E(G_k)|\le O(|E|\cdot \log n)$.
Also, we spend time only on relevant scales, and the number of relevant scales is at most $K = O(|E| \cdot \log n)$,
 By \claimref{claim:active-nodes}, we conclude that the total expected running time is
\[
\sum_{k \in K} O(|E(G_k)| + n_k \cdot \log n) \cdot n^\rho/\rho ~=~ O(|E|+n\log n)\cdot n^\rho/\rho\cdot\log n~.
\]
We thus have the following theorem.
\begin{theorem}\label{thm:standard-hop}
For any graph $G = (V,E,\omega)$ with $n$ vertices, $2\le\kappa\le (\log n)/4$, $1/2 > \rho \ge 1/\kappa$, and $0 < \eps < 1/2$,
our algorithm constructs a $(\beta,\eps)$-hopset $H$ with $O(n^{1+1/\kappa} \cdot \log n)$ edges in expectation,
in expected  time $O((|E| + n \log n) (n^\rho/\rho\cdot \log n) )$, with
$\beta =O\left( {{\log \kappa + 1/\rho} \over \eps} \right)^{\log \kappa + 1/\rho}$.
\end{theorem}

\subsection{Implementation in the Streaming Model}\label{sec:reduce-stream}

We assign $O(n\log n)$ words of memory for storing a data structure for nodes of the graphs $\{G_k\}_{k\in K}$. The main observation is that whenever we contract an edge between nodes $X,Y$ with $|X|\ge |Y|$, only the vertices of $Y$ get a new center, but the size of the node containing them is  at least doubled. This implies that each vertex changes the center of the node containing it at most $\log n$ times. Every $x\in V$ stores a list $L(x)$ of pairs, where a pair $(i,v)\in L(x)$ indicates that at scale $i\in K$ the node containing $x$ was merged with a larger node centered at $v$. Initially, $(0,x)\in L(x)$.
The lists $\Lists = \{L(v) \mid v \in V\}$ that our algorithm maintains enable us to maintain a part $\cL'$ of the laminar family  $\cL$ that was constructed so far.
Observe that $\cL'$ can be viewed as a forest of sets, and this forest is partial to the tree $\cL$.
The nodes of each $G_k$ can be reproduced from the lists $\{L(x)\}_{x\in V}$.
 (Specifically, to compute the nodes of $G_k$: for each vertex $x\in V$, find the maximum index $i\le k$, for which there is an entry $(i,v)\in L(x)$, and $x$ will be a part of a $G_k$-node centered at $v$.)

\begin{comment}
We now describe how to generate this data structure in one pass over the stream. We stress that only the nodes of $G_k$ can be computed, clearly we cannot afford to store its edges. The main obstacle is that edges do not come in sorted order. Whenever an edge $(x,y)\in E$ of weight $\omega(x,y)$ is read from the stream, let $k\in K$ be the minimal such that $\omega(x,y)<(\eps/n)\cdot 2^k$. We will need to merge the nodes of $x$ and $y$ at scale $k$ (if they aren't merged already), and update the names of the centers of all their ancestor nodes (because the nodes' size changed due to this merge). This corresponds to merging the branches in the partition tree of the paths to the nodes containing $x$ and $y$ into a single branch (which may have a new node at scale $k$), and possibly renaming some of the nodes' centers. Clearly this could be done in linear space. We then update the data structure of all vertices according to the new partition tree. Once the data structure is computed, we can insert all the edges of $S$ to the hopset, since these edges depend only on the structure of the nodes.
\end{comment}

Therefore, the algorithm does not really need to store the list of the edges that it has seen so far. Rather the information stored in $\Lists$, along with the new edge $e = (x,y)$, which the algorithm processes in the current stage, is sufficient for updating the set $\Lists$, and the latter is sufficient for deducing the node sets of graphs $\{G_k\}_{k \in K}$. More concretely, given the set $\Lists_h$, which was constructed from a sequence $(e_1,e_2,\ldots,e_h)$ of edges, for some positive integer $h$, and a newly arriving edge $e = e_{h+1}$ of weight $\omega = \omega(e)$, the algorithm constructs the set $\Lists_{h+1}$, which reflects the appended edge sequence $(e_1,e_2,\ldots,e_{h+1})$, in the following way. It processes the scales $k = 0,1,2\ldots$ bottom-up, one after another. (Recall that we are not limited in processing time, but rather only in memory.) Initially, $\Lists_{h+1}$ is empty.
In each scale $k$, the algorithm processes all scale-$k$ merges recorded in $\Lists_h$, and, if the newly arrived edge $e$ causes a $k$-scale merge, then the algorithm processes this new edge as well.  We also recompute the set $S_{h+1}$ of hopset edges in each iteration. For every node $Z$, these edges connect the node center $z$ with every other vertex $v \in Z$.  That is, the previous set $S = S_h$ is discarded, and the new set $S = S_{h+1}$ is computed from scratch.

This completes the description of the first pass of our algorithm, i.e., of the pass that computes the set $\Lists$, and as a result, the node sets of graphs $G_k$, for all relevant scales $k$. The space required for this computation is proportional to the maximum size of the data structure $\Lists$, which is, by Claim \ref{claim:active-nodes}, at most $O(n \cdot \log n)$.

The correctness of this procedure hinges on the observation that if two nodes $X,Y$ merge on scale $k$, it is immaterial which of the edges from $(X \times Y) \cap E$ caused this merge. Moreover, the weight $\omega(e)$ of this edge is also immaterial. (By the very fact that the merge occurred, we know that $\omega(e) < {\eps \over n} \cdot 2^k$.)

Let us now review the execution of the hopset algorithm in the consequent passes over the stream. We shall compute single scale hopsets $H_k$ in parallel for all $k\in K$. For each $k\in K$ we run the hopset construction given by \theoremref{thm:stream_hop-alt}. Initially, the vertices of each $G_k$ can be derived from the data structure we store.
Whenever an edge $(x,y)\in E$ of weight $\omega(x,y)$ is read from the stream, we know it is active in at most $\log(n/\eps)+2$ different scales. For each such scale $k\in K$ with $(\eps/n)\cdot 2^k\le \omega(x,y)<2^{k+2}$, we use the data structure to find the centers of nodes containing $x,y$ in $G_k$, and execute the hopset algorithm as if an edge connecting these centers (of weight given by \eqref{eq:edge-set}) was just read from the stream.

We have two possible tradeoffs between space and number of passes for given parameters $\kappa$, $\eps$, $\rho$ and $t$.  Since we run in parallel, the fact that there are many graphs does not affect the number of passes. The size of each $H_k$ is only $O(n_k^{1+1/\kappa})$. Using the fact that each $G_k$ has aspect ratio at most $\Lambda_k=O(n/\eps)$, we can essentially replace $\log\Lambda$ by $\log (n/\eps) = O(\log n)$ in \theoremref{thm:stream_hop-alt}.
The total space used by the algorithm
is $\sum_{k\in K}n_k^{1+1/\kappa} \cdot \log (n/\eps) \le O(n^{1+1/\kappa}\log^2  n)$,
 rather than $O(n^{1+1/\kappa}\log\Lambda)$.
Formally, we derive the following theorem.
\begin{theorem}
\label{thm:stream_hop-alt1}
For any graph $G = (V,E,\omega)$ with $n$ vertices, any $2\le\kappa\le(\log n)/4$, $1/2 > \rho \ge 1/\kappa$,  $0 < \eps < 1/2$, and any $1\le t\le O(\log n)$,
our streaming algorithm computes a $(\beta,\eps)$-hopset
 with expected size $O(n^{1+1/\kappa} \cdot \log n)$, and with $\beta$ given by
\begin{equation}
\label{eq:reduce_beta}
\beta =  O\left( {{(\log \kappa + 1/\rho) \log n}  \over {\eps \cdot t}}\right)^{\log \kappa + 1/\rho}~.
\end{equation}
The resource usage is either
\begin{enumerate}
\item
expected space $O(t\cdot n^{1+\rho}/\rho+ n^{1+1/\kappa} \cdot \log^2 n))$
(resp., space $O(t\cdot n^{1+\rho} /\rho + n^{1+1/\kappa} \cdot \beta\cdot 2^t\cdot \log^2 n)$ for path-reporting)
and $O(\beta \cdot \log n\cdot 2^t)$ passes, whp, or
\item
expected space $O(n^{1+1/\kappa} \cdot \log^2 n)$
(resp.,  $O(n^{1+1/\kappa} \cdot  \beta \cdot 2^t\cdot \log^2 n)$ for path-reporting)
and $O(n^\rho \cdot \beta \cdot \log^2n\cdot 2^t)$ passes, whp.
\end{enumerate}
\end{theorem}

A few possible tradeoffs summarized in Table \ref{fig:table1}.
 The first four results in the table are obtained by choosing $t=1$ in Theorem~\ref{thm:stream_hop-alt1}. The first two follow from the first item, and the following two from the second item. The fifth and sixth results also follow from the second item of Theorem~\ref{thm:stream_hop-alt1}, but have an improved, i.e., independent of $n$,  $\beta$. In the fifth one we set $t=\mu \cdot \rho\cdot\log n$, for an arbitrarily small constant $\mu > 0$. Then  we rescale $\rho'=(1 +\mu) \cdot \rho$.
This induces the term of $1 + \mu$ in the exponent of $\beta$. To get the path-reporting version of this bound, in the sixth result we set a smaller $t= \mu \cdot (\log n)/k$,  and rescale $\kappa' = {\kappa \over {1 + \mu}}$.
The $O$-notation of the $\beta$-column in lines 5 and 6 of Table \ref{fig:table1} hides a (constant) factor of $1/\mu$ in the base of the exponent.

\begin{table}[!ht]
\begin{center}
\begin{tabular}{|c|c|c|c|}
	\hline
                  Space   &   $\#$ of passes & the hopbound $\beta$ & Paths \\
	\noalign{\global\arrayrulewidth0.05cm}
 \hline
 \noalign{\global\arrayrulewidth0.4pt}
$O(n^{1+\rho}/\rho+n^{1/\kappa}\log^2 n )$   &$O(\beta \log  n)$ & $O\left({{\log  n} \over {\eps }}(\log(\rho\kappa) +\frac{1}{\rho})\right)^{\log (\rho\kappa)  + \frac{1}{\rho}}$ & No \\
	\hline
$O(n^{1+\rho}  /\rho + n^{1+\frac{1}{\kappa}} \beta\log^2 n)$   &$O(\beta \log  n)$ & $O\left({{\log  n} \over {\eps }}(\log(\rho\kappa) +\frac{1}{\rho})\right)^{\log (\rho\kappa)  + \frac{1}{\rho}}$ & Yes \\
	\hline
$O(n^{1+\frac{1}{\kappa}} \cdot \log^2  n)$   & $O(n^\rho \cdot \beta \cdot \log^2 n)$ & $O\left({{\log   n} \over {\eps }}(\log(\rho\kappa) +\frac{1}{\rho})\right)^{\log (\rho\kappa)  + \frac{1}{\rho}}$ & No \\
	\hline
$O(n^{1+\frac{1}{\kappa}} \cdot \log^2  n\cdot\beta)$   & $O(n^\rho \cdot \beta \cdot \log^2 n)$ & $O\left({{\log   n} \over {\eps }}(\log(\rho\kappa) +\frac{1}{\rho})\right)^{\log (\rho\kappa)  + \frac{1}{\rho}}$ & Yes\\
	\hline
$O(n^{1+\frac{1}{\kappa}}\cdot \log^2  n)$   & $O(n^\rho \cdot \beta \cdot \log^2 n )$ & $O\left({{1} \over {\eps \cdot \rho}}(\log(\rho\kappa) +\frac{1}{\rho})\right)^{\log (\rho\kappa)  + \frac{1+\mu}{\rho}}$ & No \\
	\hline
$O(n^{1+\frac{1}{\kappa}} \cdot    \log^2  n\cdot \beta)$   & $O(n^{\frac{\mu}{\kappa}+\rho} \beta\log^2 n )$ & $O\left({{\kappa} \over {\eps}}(\log(\rho\kappa) +\frac{1}{\rho})\right)^{\log (\rho\kappa)  + \frac{1}{\rho}}$ & Yes \\
	\hline
\end{tabular}
\end{center}
\caption[]{
\label{fig:table1}
Summary of results for $(\beta,\eps)$-hopsets in the streaming model, all are with expected size $O(n^{1+1/\kappa}\log n)$ and stretch $1+\eps$.
The space bounds are in expectation, the bounds on the number of passes hold whp, and the bounds on $\beta$ hold deterministically.
The last column indicates whether the hopset is path-reporting or not. }
\end{table}

The following corollary summarizes the fifth and sizth lines of Table \ref{fig:table1}, which provide efficient (requiring roughly $\tO(n^\rho)$ time) streaming algorithms for constructing  $(\beta,\eps)$-hopsets
with $\tO(n^{1+1/\kappa})$ edges, and with $\beta = \beta(\eps,\kappa,\rho)$ independent of $n$. All the four parameters $\eps$, $1/\kappa$, $\rho$ and $\mu$ can simultaneously be made arbitrarily close to 0 constants, while still having constant $\beta$.

\begin{corollary}
\label{cor:str_hopset}
For any graph $G = (V,E,\omega)$ with $n$ vertices, any $2\le\kappa\le(\log n)/4$, $1/2 > \rho \ge 1/\kappa$,  $0 < \eps < 1/2$, and any arbitrarily small
constant $\mu > 0$,
our streaming algorithm computes an $(\beta,\eps)$-hopset (resp., path-reporting hopset)
 with expected size $O(n^{1+1/\kappa} \cdot \log n)$, and with $\beta = O\left({{1} \over {\eps \cdot \rho}}(\log(\rho\kappa) +\frac{1}{\rho})\right)^{\log (\rho\kappa)  + \frac{1+\mu}{\rho}}$ (resp., $\beta = O\left({{\kappa} \over {\eps}}(\log(\rho\kappa) +\frac{1}{\rho})\right)^{\log (\rho\kappa)  + \frac{1}{\rho}}$), expected space
$O(n^{1+{1 \over \kappa}} \cdot \log^2 n)$
(resp., $O(n^{1+{ {1} \over \kappa}} \cdot \log^2 n \cdot \beta)$),
in $O(n^\rho \cdot \beta \cdot \log^2 n)$ passes (resp.,
 $O(n^{\rho + {\mu  \over \kappa}} \cdot \beta \cdot \log^2 n)$, whp.
\end{corollary}

\subsection{Implementation in the Congested Clique Model}

In this model, the computation of the nodes of the graphs $\{G_k\}_{k\in K}$ could be done in $O(\log n)$ rounds. As commonly accepted, we shall assume that any edge weight can be sent in a single message. Let $r\in V$ be a designated root. We shall imitate the Boruvka algorithm in order to build the nodes of $G_k$ -- in each iteration, for every node $U$, every vertex $u\in U$ will send to the root $r$ the lightest edge incident on it that leaves $U$. Then $r$ will locally compute the new nodes created by merging all of the reported edges, and will also take note of the finer structure of the laminar family $\cL$, emerging since the edge weights could be from different scales. The root will send to each vertex $x\in V$, the name of the center of the current node containing it, and the scale in which $x$ joined this node. Between iterations, every vertex will send the center of the node containing it to all of its neighbors, so that every vertex will know which edges are leaving its current node. Since after iteration $i$ the size of each node is at least $2^i$, after $\log n$ iterations the process ends. In every iteration we sent 2 messages over each edge, so the number of rounds required is $O(\log n)$.
At the end of the process, the root will compute and broadcast the lists $L(x)$ for each vertex $x\in V$ as in \sectionref{sec:reduce-stream}.
Each $L(x)$ consists of $O(\log n)$ words. This will require additional $O(\log n)$ rounds, which can be verified by noting that we can send $L(x)$ to $x$, and then all vertices $x$  will in parallel send their sets $L(x)$ to all other vertices.

Finally, the root will appoint a coordinator vertex $c(U)\in U$ for every node $U$, who will be in charge of communications for that node. We  require that every vertex participates at most once as a coordinator for a non-trivial node (a node of size $>1$). Note that when we create a new node $U$ by combining two others $X$, $Y$, we can maintain the property that there will be a vertex $u\in U$ who never was a coordinator -- this holds by induction for $x\in X$ and $y\in Y$, so we can simply set $x$ as the coordinator for $U$ and $u=y$ will be the "free" vertex. (We do not use the center as coordinator, since the same vertex can be a center in numerous scales.)

Similarly to the case of the streaming model, we run the hopset algorithm of \sectionref{sec:hopset_clique}  for all graphs $G_k$ in parallel. Let us review briefly how to implement each step in the algorithm. Recall that the main ingredients are Bellman-Ford explorations, in every iteration of which, every vertex sends its current estimate to all of its neighbors. In the graph $G_k$, for every node $U$ the coordinator $c(U)$ will send the appropriate estimate $\hat{d}$ to all the vertices in the graph (along with the  scale $k$ of the node $U$). Every vertex $u\in V$ that receives this message, sends to its coordinator at level $k$ the updated estimate $\hat{d}+\hat{\omega}$, where $\hat{\omega}$ is the shortest length of an edge (given by \eqref{eq:edge-set}) connecting $u$ to a vertex in $U$. (Recall that each vertex knows the entire laminar family $\cL$ of sets, and hence can compute locally the nodes of $G_k$.) The coordinator of each node will keep the shortest of these as the estimate for its node. Thus, for each step of Bellman-Ford, we need two rounds, in which all communication is over edges containing a coordinator.

We now analyze the required number of rounds. We charge the cost of each exploration step to the coordinators of the nodes. The point is that every vertex can be a coordinator in at most $\log(n/\eps)+2$ different scales $k\in K$, because once a node is active, after so many scales it must be merged with another node, which will necessarily have a different coordinator. We conclude that the load on any edge, arising from it participating in many different graphs, is only $O(\log n)$.
Hence  the number of rounds as a result of this simulation grows only by a factor of $O(\log n)$. Also recall that the aspect ratio of each $G_k$ is $O(n/\eps)$. By applying the single-scale versions of \theoremref{thm:impr_distr_hop} and Theorem \ref{thm:clique_path_rep_hopset}  on each $G_k$, we conclude with the following.

\begin{theorem}
\label{thm:congest-clique-alt}
For any graph $G = (V,E,\omega)$ with $n$ vertices, $2\le\kappa\le (\log n)/4$, $1/2 > \rho \ge 1/\kappa$, $1\le t\le\log n$, and $0 < \eps < 1/2$, our distributed algorithm for the Congested Clique model computes a $(\beta,\eps)$-hopset $H$ with expected size $O(n^{1+1/\kappa} \cdot \log n)$
in $O(n^\rho/\rho\cdot\log^3n\cdot\beta\cdot 2^t)$ rounds whp, with $\beta$ given by \eqref{eq:reduce_beta}.
For a path-reporting hopset, the number of rounds becomes larger by a factor of $\beta\cdot 2^t$.
\end{theorem}

To get a hopset with $\beta$ independent of $n$, we set $t = \log n^{\mu\rho}$, for an arbitrarily small constant $\mu > 0$. We then rescale $\rho' = (1+\mu)\rho$.
As a result, we obtain
\begin{equation}
\label{eq:beta_distr_impr}
\beta ~=~ O\left({1 \over {\eps \cdot \rho}} (\log \kappa + 1/\rho)\right)^{\log \kappa+ {{1 + \mu} \over \rho}}~.
\end{equation}
The $O$-notation in (\ref{eq:beta_distr_impr}) hides a (constant) factor of $1/\mu$ in the base of the exponent.
The number of rounds becomes, whp, $O(n^\rho/\rho \cdot \log^3 n \cdot \beta)$.
In the  path-reporting case, we set $t = \log n^{(\mu/2) \rho}$, and rescale in the same way as above.
As a result, the running time becomes $O(n^\rho/\rho \cdot \log^3 n \cdot \beta^2)$, whp. The hopbound and the hopset size are the same as in the not path-reporting case.

\begin{corollary}
\label{cor:congest_clique_reduction}
For any graph $G = (V,E,\omega)$ with $n$ vertices, $2\le\kappa\le (\log n)/4$, $1/2 > \rho \ge 1/\kappa$,  and $0 < \eps < 1/2$, and any constant $\mu > 0$, our distributed algorithm for the Congested Clique model computes a $(\beta,\eps)$-hopset $H$ with expected size $O(n^{1+1/\kappa} \cdot \log n)$
in $O(n^\rho/\rho\cdot\log^3n\cdot\beta)$ rounds whp, with $\beta$ given by \eqref{eq:beta_distr_impr}.
For a path-reporting hopset, the number of rounds is $O(n^\rho/\rho \cdot \log^3 n \cdot \beta^2)$, whp.
\end{corollary}

\subsection{Implementation in the CONGEST Model}

In this model we are given a "virtual" graph $\tG=(\tV,\tE,\tomega)$,  where $\tV\subseteq V$, $|\tV| = m$, on which we wish to compute a hopset. For constructing the nodes of the graphs $G_k$ we do essentially the same as we did in the Congested Clique model, with two small differences: Replace every message sent to/from the root, by a broadcast to every vertex in the graph, and also we do not need each vertex to notify its neighbors of its current node, since this information was sent to the entire graph already. Denote by $\tV_k$ their vertex sets. For every scale $k$,
in each iteration of the Bellman-Ford algorithm in $\tG_k$, we send $O(|\tV_k|)$ messages. So the total number of messages sent is
$M=O(\sum_{k \in K} |\tV_k|) = O(|\tV|  \log m) = O(m \cdot \log m)$. These $M$ messages  can be convergecasted and broadcasted over the BFS tree of $G$ in $O(M+D) = O(D + m \cdot \log m)$ rounds, where $D$ is the hop-diameter of $G$.

The hopset algorithm from \sectionref{sec:congest} will require a subtle modification: When a coordinator broadcasts a distance estimate $\hat{d}$, we cannot afford to have every vertex in every node notify its coordinator of its own estimate. Rather than that, we convergecast the information on the global BFS tree, while forwarding at most one message per node. More formally, for every node $U$ that has a vertex who received a distance estimate from the coordinator of some other node $U'$, only a single message will be sent up in the BFS tree -- the one with minimal estimate. Then the root will broadcast the updated distance estimates to all the coordinators. If the total number of distance updates required is $M$, it will require  only $O(M+D)$ rounds, see, e.g., \cite[Lemma~3.4.6]{P00}.

By Theorems \ref{thm:congest} and \ref{thm:hop_aware_small_beta}, the number of rounds required for computing  all $t$ hopsets of a given group of  scales is
$O((D + m^{1+\rho} \cdot \log m \cdot t) \beta/\rho \cdot 2^t )$ in the not path-reporting case, and is
$O((D + m^{1+ \rho} \cdot \log m \cdot t \cdot \beta \cdot 2^t) \beta/\rho \cdot 2^t )$ in the path-reporting one. (Both bounds are whp.)
To get a hopset for all scales, this expression was multiplied by the number of groups, i.e., $\lceil {{\log \Lambda} \over t} \rceil$.
When computing the hopsets $H_k$ of $G_k$, the number of groups is $\lceil {{\log O(m/\eps)} \over t} \rceil = O({{\log m} \over t})$, because the aspect ratio of each $G_k$ is $O(m/\eps)$. The number of messages convergecasted and broadcasted over the BFS tree $\tau$ of the entire network in each  iteration of the Bellman-Ford algorithm (which is executed now in parallel in all graphs $\tG_k$) is $\sum_{k \in K} O(m_k^{1+ \rho} \log m_k \cdot t) = O(m \cdot \log^2 m \cdot t)$, whp, in the not path-reporting case, and is $\sum_{k \in K} O(m_k^{1+\rho} \cdot \log m_k \cdot t \cdot \beta \cdot 2^t) = O(m \cdot \log^2 m \cdot t \cdot \beta \cdot 2^t)$, whp, in the path-reporting one.

We summarize this discussion with the following theorem.

\begin{theorem}
\label{thm:congest-alt}
For any graph $G = (V,E)$ with hop-diameter $D$, and any $m$-vertex weighted graph $\tG= (\tV,\tE,\tomega)$ embedded in $G$,  and
any $2\le\kappa\le (\log m)/4$, $1/2 > \rho \ge 1/\kappa$, $1\le t\le\log m$, $0 < \eps < 1/2$,  our distributed algorithm in the CONGEST model computes a $(\beta,\eps)$-hopset $H$ for $\tG$ with expected size $O(m^{1+1/\kappa} \cdot \log m)$,
in $O((D+ m^{1+\rho} \cdot \log^2 m\cdot t)\cdot\beta/\rho\cdot \log m\cdot 2^t/t)$ rounds whp, with $\beta$ given by \eqref{eq:reduce_beta}, with $n$ replaced by $m$.
For a path-reporting hopset, the number of rounds becomes $O((D+ m^{1+\rho} \cdot \log^2 m\cdot t\cdot\beta\cdot 2^t) \cdot \beta/\rho \cdot\log m\cdot 2^t/t)$.
\end{theorem}

To get $\beta$ independent of $m$, we set $t = \log m^{\mu \rho}$.
As a result we get
\begin{equation}
\label{eq:beta_congest_const}
\beta ~=~ O\left({1 \over {\eps \rho}} (\log \kappa+ 1/\rho)\right)^{\log \kappa+ {{1} \over \rho}}~.
\end{equation}
The running time is $O((D + m^{1+ \rho} \cdot \log^3 m \cdot \rho) \cdot \beta/\rho^2 \cdot m^{\mu \rho})$, whp, in the not path-reporting case.

In the path-reporting case we get time
$O((D + m^{\rho(1 + \mu)} \cdot \log^3 m \cdot \rho \cdot \beta) \cdot {\beta \over {\rho^2}} \cdot m^{\mu \rho})$, whp. By rescaling $\rho' = \rho(1+ \mu)$, we get
\begin{equation}
\label{eq:beta_congest_const_pr}
\beta ~=~ O\left({1 \over {\eps \rho}} (\log \kappa+ 1/\rho)\right)^{\log \kappa+ {{1+\mu} \over \rho}}~,
\end{equation}
and the running time is
$O((D + m^\rho \cdot \log^3 m \cdot \rho \cdot \beta) \cdot {\beta \over {\rho^2}} \cdot m^{\mu \rho})$, whp.

\begin{corollary}
\label{cor:congest_hopset_reduction}
For any graph $G = (V,E)$ with hop-diameter $D$,
and any $m$-vertex weighted graph $\tG= (\tV,\tE,\tomega)$ embedded in $G$,
any
 $2\le\kappa\le (\log m)/4$, $1/2 > \rho \ge 1/\kappa$,  and $0 < \eps < 1/2$, and any constant $\mu > 0$,   and
our distributed algorithm for the CONGEST model computes a $(\beta,\eps)$-hopset $H$ for $\tG$ with expected size $O(m^{1+1/\kappa} \cdot \log m)$,
in $O((D+ m^{1+\rho} \cdot \log^3 m \cdot \rho)\cdot\beta/\rho^2  \cdot m^{\mu\rho})$ rounds whp, with $\beta$ given by \eqref{eq:beta_congest_const}.
For a path-reporting hopset, $\beta$ is given by (\ref{eq:beta_congest_const_pr}),
the number of rounds becomes $O((D+ m^{1+\rho} \cdot \log^3 m\cdot \beta \cdot \rho) \cdot m^{\rho \mu} \cdot \beta/\rho^2)$.
\end{corollary}

\subsection{PRAM Model}
\label{sec:pram_reduction}

Klein and Sairam \cite{KS93}  (see also \cite{C97}) showed  that the graphs $G_1,G_2,\ldots,G_\lambda$, $\lambda=\lceil\log\Lambda\rceil$, can be computed in EREW PRAM model in $O(\log^2 n)$ time, using $O(|E|)$ processors. We next compute hopsets $H_k$, for all $k \in K$, in parallel. The overall expected size of the resulting hopset is, by Theorem \ref{thm:pram_gen_hopset},
$\Expect(|H|) = \sum_{k \in K} n_k^{1+ 1/\kappa} = O(n^{1+ 1/\kappa} \cdot \log n)$.
The aspect ratio of each graph $G_k$ is $O(n/\eps)$, and thus, the number of processors used is
$$
O(\sum_{k \in K} (|E(G_k)| + n_k^{1+1/\kappa} \cdot \log^2 n) \cdot n^\rho \cdot \log n \cdot t) ~=~ O(n^\rho \cdot \log n \cdot t \cdot   (|E| \cdot \log n + n^{1+1/\kappa} \cdot \log^3 n))~.$$
To summarize:
\begin{theorem}
\label{thm:pram_reduction}
For any $n$-vertex graph $G = (V,E,\omega)$ of diameter $\Lambda$, any $2\le\kappa\le(\log n)/4$, $1/2 > \rho \ge 1/\kappa$,  $0 < \eps \le 1$, and any $1 \le t \le \log \Lambda$,
our parallel  algorithm computes a $(\beta,\eps)$-hopset with $O(n^{1+1/\kappa} \cdot \log n)$ edges in expectation, and with $\beta$ given by
$$\beta ~= ~ O\left({{(\log \kappa + 1/\rho) \cdot \log n} \over {\eps \cdot t}}\right)^{\log\kappa +1/\rho}~,$$
in $O(\beta \cdot (\log \kappa+ 1/\rho) \cdot \log^2 n \cdot {{2^t} \over t})$ EREW PRAM time, using
$O((|E| + n^{1+1/\kappa} \cdot \log^2 n) \cdot n^\rho \cdot \log^2 n \cdot t)$ processors, whp.
\end{theorem}

In particular, by setting $t = 1$ we get
\begin{equation}
\label{eq:pram_beta_reduction}
\beta ~=~ O\left({{(\log \kappa + 1/\rho) \cdot \log n} \over \eps}\right)^{\log \kappa + 1/\rho}~,
\end{equation}
in $O(\beta \cdot (\log \kappa + 1/\rho) \cdot \log^2 n)$ EREW PRAM time, using $O((|E| + n^{1+1/\kappa} \cdot \log^2 n) \cdot n^\rho \cdot \log^2 n)$ processors.
One can also set $t = \log n^\zeta$, for a parameter $\zeta > 0$, and get
$O(\beta \cdot (\log \kappa + 1/\rho) \cdot \log n \cdot n^\zeta/\zeta)$ time, with
$$\beta = O\left({{\log \kappa + 1/\rho} \over {\eps \cdot \zeta}} \right)^{\log \kappa + 1/\rho}~.$$
Note that $\beta$ becomes independent of $n$, i.e., it is constant whenever $\rho$, $\eps$, $\zeta$ and $1/\kappa$ are.
So one can compute a $(\beta,\eps)$-hopset, with arbitrarily small constant $\eps > 0$, in EREW PRAM time $O(n^\zeta)$, for an arbitrarily small constant $\zeta > 0$,
using $O((|E|  + n^{1+ 1/\kappa} \cdot \log^2 n) \cdot n^\rho \cdot \log^3 n)$ processors, for arbitrarily small constants $\rho,1/\kappa > 0$, and still have a constant hopbound  $\beta$.

\section{Applications}
\label{sec:apps}

In this section we describe applications of our improved constructions of hopsets to computing approximate shortest paths for a set $S \times V$ of vertex pairs, for a subset $S \subseteq V$ of designated sources.

\subsection{Congested Clique Model}

Let $G=(V,E,\omega)$ be a weighted graph with $n$ vertices, $0<\eps<1/2$, and let $S\subseteq V$ be a set of $s=|S|$ sources. In order to compute shortest paths from every vertex in $S$ to every vertex in $V$, we first apply \corollaryref{cor:congest_clique_reduction} to the graph with parameters $\kappa= \log_sn$ and $\rho=1/\kappa$ (we assume $s\ge 16$ for the bound on $\kappa$ to hold). We use $\mu = 0.01$. This yields a $(\beta,\eps)$-hopset $H$ with $\beta=O((\log_sn)/\eps)^{2.02\cdot\log_sn}$. Now each of the $S$ sources (in parallel)  conducts $\beta$ iterations of Bellman-Ford exploration in $G\cup H$, and as a result obtains $1+\eps$ approximate distance estimations to all other vertices.

The number of rounds required to compute the hopset is whp \\
$O(n^\rho\cdot\log^4n\cdot\beta)=s\cdot O((\log_sn)/\eps)^{2.02\cdot\log_sn}\cdot\log^4n$, and the number of rounds to conduct $s$ Bellman-Ford explorations to range $\beta$ is at most $O(s\cdot\beta)$ (see, e.g., \sectionref{sec:hopset_clique}). We conclude that the total number of rounds is $s\cdot O((\log_sn)/\eps)^{2.02\cdot\log_sn}\cdot\log^4n$.
In the case that $s=n^{\Omega(1)}$, we can in fact set $\kappa=\log_{s/\lg^4 n}n$, which yields $\beta=(1/\eps)^{O(1)}$, and the number of rounds will be essentially linear in $s$, specifically,  $s\cdot(1/\eps)^{O(1)}$.
In the case that $s\le 2^{\sqrt{\log n\log\log n}}$, it is more beneficial to choose $\kappa=\sqrt{{{\log n} \over {\log\log n}}}$. This  yields
\begin{equation}
\label{eq:beta_appl}
\beta=(1/\eps)^{O(\sqrt{{\log n} \over {\log\log n}})} \cdot 2^{O(\sqrt{\log n \cdot \log\log n})}~,
\end{equation}
 in $\tO(n^\rho \cdot \beta) = O(1/\eps)^{\sqrt{{\log n} \over {\log\log n}}} \cdot 2^{O(\sqrt{\log n \log\log n})}$ rounds.

If one is interested in the actual paths, rather than just distances, the we employ our path-reporting variant of hopsets.  This increases the number of rounds by an additional factor of $\beta$.

\begin{theorem}\label{thm:CC-paths}
For any graph $G=(V,E,\omega)$ with $n$ vertices, a parameter $0<\eps<1/2$, and $S\subseteq V$ of size $s$, there is a  algorithm in the Congested Clique model, that whp computes $(1+\eps)$-approximate $S\times V$ shortest paths in $s\cdot O((\log_sn)/\eps)^{2.02\cdot\log_sn}\cdot\log^4n$ rounds. In the case  $s=n^{\Omega(1)}$ we can achieve  $s\cdot(1/\eps)^{O(1)}$ rounds, and in the case $s\le 2^{\sqrt{\log n \log\log n}}$ the number of rounds can be made
 $s \cdot (1/\eps)^{O(\sqrt{{\log n} \over {\log\log n}})} \cdot 2^{O(\sqrt{\log n \log\log n})}$.
\end{theorem}

\subsection{CONGEST Model}

Computing approximate shortest paths in the CONGEST model using hopsets is somewhat more involved. We shall follow the method of \cite{HKN15}, and give full details for completeness.
Let $G=(V,E,\omega)$ be a weighted graph with $n$ vertices, $0<\eps<1/2$, and let $S\subseteq V$ be a set of $s=|S|$ sources.
%
First we quote a Lemma of \cite{N14}, which efficiently computes hop-limited distances from a given set of sources.
\begin{lemma}[\cite{N14}]\label{lem:N14}
Given a weighted graph $G=(V,E,\omega)$ of hop-diameter $D$, a set $\tV\subseteq V$, and parameters $t\ge 1$ and $0<\eps<1/2$, there is a distributed algorithm that whp runs in $\tilde{O}(|\tV|+t+D)/\eps$ rounds, so that every $u\in V$ will know values $\{\td(u,v)\}_{v\in \tV}$ satisfying\footnote{The computed values are symmetric, that is, $\td(u,v)=\td(v,u)$ whenever $u,v\in \tV$.}
\begin{equation}\label{eq:duv}
d_G^{(t)}(u,v)\le \td(u,v)\le (1+\eps)d_G^{(t)}(u,v)~.
\end{equation}
\end{lemma}
\begin{remark}\label{rem:parents}
While not explicitly stated in \cite{N14}, the proof also yields that each $v\in V$ knows, for every $u\in \tV$, a parent $p=p_u(v)$ which is a neighbor of $v$ satisfying
\begin{equation}\label{eq:p-u}
\td(v,u)\le \omega(v,p)+\td(p,u)~.
\end{equation}
\end{remark}

Let $\tV\subseteq V$ be a random set of vertices, such that each $v\in V$ is included in $\tV$ independently with probability $1/\sqrt{sn}$. Note that whp $|\tV|\le \sqrt{n/s}\cdot\ln n$, so that $s\cdot|\tV|=\tilde{O}(\sqrt{ns})$.
The following claim argues that the random sample $\tV$ hits every shortest path somewhere in its first $\tilde{O}(\sqrt{sn})$ vertices.
\begin{claim}\label{claim:on-path}
The following holds whp: for every $x,y\in V$, there exists $u\in\tV\cup\{y\}$ on the shortest path from $x$ to $y$ in $G$, such that $d_G^{(4\sqrt{sn}\cdot\ln n)}(x,u)=d_G(x,u)$.
\end{claim}
\proof
Fix some $x,y\in V$. If it is the case that the shortest path between them $\pi(x,y)$ in $G$ is comprised of at most $4\sqrt{sn}\cdot\ln n$ vertices, then we can take $u=y$. Otherwise, the probability that none of the first $4\sqrt{sn}\cdot\ln n$ vertices on $\pi(x,y)$ is sampled to $\tV$ is bounded by $(1-1/\sqrt{sn})^{4\sqrt{sn}\cdot\ln n}\le n^{-4}$. Taking a union bound on the $O(n^2)$ pairs concludes the proof.
\QED

Let $\tG=(\tV,\tE)$ be the graph on the vertex set $\tV$ of size $m=|\tV|$, with edge weights $\td(u,v)$ given by applying \lemmaref{lem:N14} on $G$ with parameters $t=4\sqrt{sn}\cdot\ln n$ and $\eps$. This will take $\tilde{O}((D+\sqrt{sn})/\eps)$ rounds. Next, construct a $(\beta,\eps)$-hopset $H$ for $\tG$ (embedded in $G$) as in \corollaryref{cor:congest_hopset_reduction}, with $\kappa=\sqrt{\log m/\log\log m}$, $\rho=1/\kappa$ (and, say $\mu=0.01$). This results in
\begin{equation}
\label{eq:beta_m_appl}
\beta= (1/\eps)^{O(\sqrt{{\log m} \over {\log\log m}})} \cdot 2^{O(\sqrt{\log m \cdot \log\log m})}~,
\end{equation}
and the number of rounds required is
$(D + m ) \cdot (1/\eps)^{O(\sqrt{\log m/\log\log m})} \cdot 2^{O(\sqrt{\log m \cdot \log\log m})}$.
Now, for each $s\in S$, each $u\in\tV$ holds an initial estimate $\td(u,s)$ given by \lemmaref{lem:N14}. We conduct $\beta$ iterations of Bellman-Ford explorations in $\tG\cup H$ for each of the vertices of $S$. That is, in every iteration, every $u\in\tV$ sends $s$ messages containing its current distance estimate for each $s\in S$, and updates its estimates according to the messages of other vertices.
This requires additional $O(D+ms)\cdot\beta$ rounds.
As a result, for every pair $s\in S$ and $u\in\tV$, the vertex $u$ holds an estimate $\hat{d}(s,u)$.
We broadcast all these values to the entire graph, in $O(D+sm)$ rounds.

Finally, for each $v\in V$ and $s\in S$, the vertex $v$ computes the value $\hat{d}(v,s)=\min_{u\in\tV}\{\td(v,u)+\hat{d}(u,s)\}$ as its approximate distance to $s$. The total number of rounds required is
$(D+\sqrt{ns})\cdot (1/\eps)^{O(\sqrt{{\log m} \over {\log\log m}})} \cdot 2^{O(\sqrt{\log m \cdot \log\log m})}$.
As above, whenever $s=n^{\Omega(1)}$ it is  better to set $\kappa=\log_{s^{1-\mu}/\log^4n}n$, and $\rho = 1/\kappa$. Then $\beta=(1/\eps)^{O(1)}$, and the number of rounds will be $(D+\sqrt{ns})\cdot(1/\eps)^{O(1)}$.

If one is interested in the actual paths, then we  use \remarkref{rem:parents} to trace down the parents, and the actual approximate path from any $v\in V$ to any $u\in \tV$ can be derived. Also, we shall use the path-reporting version of our hopset.  This introduces an additional factor of $\beta$ to the number of rounds, and enables every vertex in the graph to find out the actual paths that implement every hopset edge (for every hopset edge we broadcast also the path of length at most $\beta$ that implements it). In particular, $v$ will be able to infer the paths for both $\td(v,u)$ and $\hat{d}(u,s)$.

It remains to prove the correctness of the algorithm. First consider any $y\in\tV$ and $s\in S$. Let $u\in\tV$ be the vertex on $\pi(s,y)$ guaranteed by \claimref{claim:on-path} (it could be that $u=y$). By \lemmaref{lem:N14},
\begin{equation}\label{eq:ppo}
\td(s,u)\le (1+\eps)d_G^{(t)}(s,u)=(1+\eps)d_G(s,u)~.
\end{equation}
We also have that
\begin{equation}\label{eq:ppo1}
d_{\tG}(y,u)\le(1+\eps)d_G(y,u)~,
\end{equation}
where \eqref{eq:ppo1} holds because every edge along the shortest path from $y$ to $u$ was stretch in $\tG$ by at most $1+\eps$. Finally, the property of hopsets suggests that
\begin{equation}\label{eq:ppo2}
d^{(\beta)}_{\tG\cup H}(y,u)\le(1+\eps)d_{\tG}(y,u)\stackrel{\eqref{eq:ppo1}}{\le}(1+\eps)^2d_G(y,u)~.
\end{equation}
Note that in the Bellman-Ford iterations, the vertex $y$ could have heard the estimate from $u$ using a path of length $\beta$ in $\tG\cup H$. Combining \eqref{eq:ppo} and \eqref{eq:ppo2} yields that
\begin{equation}\label{eq:ppo3}
\hat{d}(y,s)\le d^{(\beta)}_{\tG\cup H}(y,u)+\td(s,u)\le(1+3\eps)d_G(y,u)+(1+\eps)d_G(s,u)\le(1+3\eps)d_G(y,s)~.
\end{equation}

Consider now some arbitrary $v\in V$ and $s\in S$.
%
By \claimref{claim:on-path}, there exists $u\in \tV\cup\{s\}$ on the shortest path from $v$ to $s$ in $G$ with $d_G^{(t)}(v,u)=d_G(v,u)$. By \lemmaref{lem:N14},
\[
\hat{d}(v,s)\le\td(v,u)+\hat{d}(u,s)\stackrel{\eqref{eq:ppo3}}{\le} (1+\eps)d_G(v,u) + (1+3\eps)d_G(u,s)\le(1+3\eps)d_G(v,s)~.
\]
We summarize by the following theorem.

\begin{theorem}\label{thm:congest-paths}
For any graph $G=(V,E,\omega)$ with $n$ vertices and hop-diameter $D$, a parameter $0<\eps<1/2$, and $S\subseteq V$ of size $s$, there is an algorithm in the CONGEST model, that whp computes $(1+\eps)$-approximate $S\times V$ shortest paths in
$O(D+\sqrt{ns})\cdot    (1/\eps)^{O(\sqrt{{\log n} \over {\log\log n}})} \cdot 2^{O(\sqrt{\log n \cdot \log\log n})}$ rounds.
Whenever $s=n^{\Omega(1)}$, we  have only $\tilde{O}(D+\sqrt{ns})\cdot(1/\eps)^{O(1)}$ rounds.
\end{theorem}

\subsection{Streaming Model}

Let $G=(V,E,\omega)$ be a weighted graph with $n$ vertices, set $0<\eps<1/2$, and let $S\subseteq V$ be a set of $s=|S|$ sources. Similarly to the Congested Clique model, we first compute a $(\beta,\eps)$-hopset $H$ for $G$, and then using additional $O(s\cdot\beta)$ passes over the stream, conduct $\beta$ iterations of Bellman-Ford exploration in $G\cup H$ separately for each of the $s$ sources. Note that each exploration requires linear in $n$ space, plus the space required to store the hopset. The latter space will be sub-linear in the output size, so we need to assume that after computing distances (or paths) from a certain source, we may output the result and erase it to free memory.

In the  case of path-reporting hopset, for each edge $e$ of the hopset, the hopset stores a path with at most $\beta$ edges $e_1,\ldots,e_\beta$ of a sub-hopset of lower scale that implements it. The same is true for each of the edges $e_1,\ldots,e_\beta$ as well, recursively. So the paths can be retrieved given  our path-reporting hopset.

We shall use the hopset given by \corollaryref{cor:str_hopset} with parameter $\kappa=(\log n)/4$.
Thus $\beta=((\log\log n+1/\rho)/\eps)^{O(\log\log n+1/\rho)}$, while for the path-reporting case $\beta_{\rm path}=((\log n)/\eps)^{O(\log\log n+1/\rho)}$. The space requirement is $O(n\log^2n)$ (for path-reporting it is larger by a factor of $\beta_{\rm path}$), and the number of passes required is $O(s\cdot\beta+n^\rho\cdot\beta\cdot\log^2n)$.

We consider several possible regimes. Whenever $s>n^{1/\log\log n}$ we set $\rho=1/\log\log n$, when $2^{\sqrt{\log n\log\log n}}<s\le n^{1/\log\log n}$ take $\rho = {{\log s} \over {\log n}}$, and for $s\le 2^{\sqrt{\log n\log\log n}}$ we choose $\rho=\sqrt{\log\log n/\log n}$. The resulting algorithms are described in the following theorem.

\begin{theorem}\label{thm:stream-paths}
For any graph $G=(V,E,\omega)$ with $n$ vertices, a parameter $0<\eps<1/2$, and $S\subseteq V$ of size $s$,  there is an algorithm in the streaming model that whp computes $(1+\eps)$-approximate $S\times V$ shortest paths, with the following resources:
\begin{itemize}
\item
Whenever $s>n^{1/\log\log n}$, the algorithm performs $s\cdot (\log n)^{O(\log^{(3)} n + \log 1/\eps)}$ passes over the stream, and use $O(n\log^2n)$ space. For path-reporting the number of passes is $s\cdot(\log n)^{O(\log\log n + \log 1/\eps)}$, while the space increases to $n\cdot(\log n)^{O(\log\log n + \log 1/\eps)}$.

\item
Whenever $2^{\sqrt{\log n\log\log n}}<s\le n^{1/\log\log n}$, the algorithm makes $s\cdot n^{O(\log\log n+\log 1/\eps)/\log s}$ passes, and the space requirement is $O(n\log^2n)$ (or $n^{1+O(\log\log n+\log 1/\eps)/\log s}$ space for path-reporting).

\item
Whenever $s\le 2^{\sqrt{\log n\log\log n}}$, the algorithm uses $2^{O(\sqrt{\log n\log\log n})}$ passes, and the space is $O(n\log^2n)$ (or $n\cdot 2^{O(\sqrt{\log n\log\log n})}$ space for path-reporting).
\end{itemize}
\end{theorem}

We also remark that whenever $s=n^{\Omega(1)}$, one can also choose a smaller $\kappa=O(1)$, increasing the space to $O(n^{1+1/\kappa}\cdot\log^2n)$, while setting $\rho=1/\kappa$ so that $\beta$ is a constant. This will yield near optimal $s\cdot(1/\eps)^{O(1)}$ passes over the stream (also for path-reporting).

\subsection{PRAM Model}
\label{sec:pram_applications}

We use Theorem \ref{thm:pram_reduction}  with $t = 1$ to construct a $(\beta,\eps)$-hopset $H$ of expected size $O(n^{1 + 1/\kappa} \cdot \log n)$ with $\beta$
given by (\ref{eq:pram_beta_reduction}), in $O(\beta \cdot (\log \kappa + 1/\rho) \cdot \log^2 n)$ parallel time, using $O((|E| + n^{1+1/\kappa} \cdot \log^2 n) \cdot n^\rho \cdot \log^2 n)$ processors.

To compute $(1+\eps)$-approximate shortest distances (or paths) for $S \times V$, for some subset $S \subseteq V$ of vertices, we now conduct $\beta$-limited Bellman-Ford explorations, in parallel, separately from each of the $|S|$ origins, in $G \cup H$.
We use $s = |S|$ processors $p_{v,1},\ldots,p_{v,s}$ for every vertex $v \in V$, and $s$ processors $p_{e,1}\ldots,p_{e,s}$, for every edge $e \in G \cup H$.
As was argued in Section \ref{sec:pram_hopset}, these explorations can now be completed in $O(\beta \cdot \log n)$ EREW PRAM time.

If we are interested in paths (rather than distance estimates), then we use a path-reporting hopset.
As a result, for every pair $(s,v) \in S \times V$, we obtain a path $\pi(s,v)$ with at most $\beta$ edges, some of which may belong to $G$, and other belong to the hopset $H$.
For each edge $e \in \pi(s,v) \cap H$, we store a path with at most $\beta$ other edges $e'$ of $G \cup H$, and the same is true for each hopset edge $e'$, recursively.
The depth of the induced tree of edges is bounded by $\log O(n/\eps)$, that is, by the aspect ratio of each of the graphs $G_k$ for which single-scale hopsets are constructed.

Hence, the entire path  can be retrieved in $O(\log n)$ parallel time, using processors that were used for conducting the Bellman-Ford explorations from vertices of $S$.
We cannot, however, retrieve all the paths simultaneously, within these resource bounds.
So our algorithm provides an implicit solution for $(1+\eps)$-approximate shortest paths problem, i.e., it returns a data structure from which each of the $S \times V$ approximate shortest paths can be efficiently extracted. Though not said explicitly, to the best of our understanding, this is also the case with Cohen's parallel $(1+\eps)$-approximate shortest paths algorithm \cite{C00} as well.

\begin{theorem}
\label{thm:pram_asp}
For any $n$-vertex graph $G = (V,E,\omega)$ of diameter $\Lambda$, a set $S\subseteq V$, and any $2\le\kappa\le(\log n)/4$, $1/2 > \rho \ge 1/\kappa$,  $0 < \eps \le 1$,
our parallel  algorithm computes  a $(1 + \eps)$-approximate shortest  distances (and an implicit $(1+\eps)$-approximate shortest paths; see above) in
$O\left( {{(\log \kappa + 1/\rho) \cdot \log n} \over \eps}\right)^{\log \kappa +1/\rho + 1} \cdot \log^2 n$ EREW PRAM time, using
$O((|E| + n^{1+1/\kappa} \cdot \log^2 n) \cdot (n^\rho \cdot \log^2 n + |S|))$ processors, whp.
\end{theorem}

Since $\rho \ge 1/\kappa$, the number of processors in Theorem \ref{thm:pram_asp} can be written as $O(|E| \cdot n^\rho \cdot \log^4 n  + |E| \cdot |S| \cdot \log^2 n)$, while incurring a term $2/\rho$ instead $1/\rho$ in the exponent of the running time. (This increase occurs as a result of rescaling $\rho' = 2\rho$.)

When $\kappa$, $\rho$ and $\eps$ are constant, this running time is polylogarithmic in $n$. In Cohen's result \cite{C00} (Theorem 1.1), the running time is polylogarithmic as well, but the exponent is roughly $O({1 \over \rho} \cdot \log \kappa)$, rather than  $O({1 \over \rho}) + \log \kappa$ in our Theorem \ref{thm:pram_asp}.

\bibliographystyle{alpha}
\bibliography{hopset}

\end{document}